\documentclass[11pt,preprint]{elsarticle}

\usepackage{graphicx}
\usepackage{rotating}
\usepackage{algorithm}
\usepackage{hyperref}
\usepackage{color}
\usepackage{amsfonts,amsmath,amssymb,amsthm}
\usepackage{bm}
\usepackage{tikz}
\usepackage{xcolor}
\usepackage{ifthen}
\usepackage{pgf}
\usepackage{comment}
\usepackage{setspace}

\newcommand{\wh}{\widehat} 		

\newcommand{\beq}{\begin{equation}}
\newcommand{\eeq}{\end{equation}}
\newcommand{\bea}{\begin{eqnarray}}
\newcommand{\eea}{\end{eqnarray}}
\newcommand{\ba}{\begin{array}}
\newcommand{\ea}{\end{array}}
\newcommand{\bi}{\begin{itemize}}
\newcommand{\ei}{\end{itemize}}
\newcommand{\ben}{\begin{enumerate}}
\newcommand{\een}{\end{enumerate}}
\newcommand{\nn}{\nonumber}

\renewcommand{\r}{\right}
\renewcommand{\l}{\left}

\theoremstyle{plain}
\newtheorem{theorem}{Theorem}
\newtheorem{lemma}{Lemma}
\newtheorem{assumption}{Assumption}
\newtheorem{corollary}{Corollary}

\begin{document}

\bibliographystyle{natbib}

\def\spacingset#1{\renewcommand{\baselinestretch}%
{#1}\small\normalsize} \spacingset{1}


  \title{\Large{\bf Sequential testing for structural stability in
approximate factor models}}
\author{Matteo Barigozzi$^1$}

\author{Lorenzo Trapani$^2$}

\fntext[fn1]{Department of Statistics, London School of Economics and Political Science.\\
Email: m.barigozzi@lse.ac.uk (corresponding author).}
\fntext[fn2]{School of Economics, The University of Nottingham\\
Email: Lorenzo.Trapani1@nottingham.ac.uk\\
\\
We wish to thank the participants to the International Conference in Memory
of Carlo Giannini (Bergamo, November 25-26, 2016); to the CMStatistics 2016
conference, (Seville, December 9-11, 2016); to the ECARES Seminar Series
(Bruxelles, February 23, 2017); to the Cass Business School Faculty of
Finance Workshop (London, March 6, 2017); to the LSE Econometrics and
Statistics seminar (London, March 10, 2017); to the Department of Economics
Seminar Series at the University of Connecticut (March 31, 2017); to the New
York Camp Econometrics XII conference (Syracuse University, April 7-9,
2017); and to the European Meeting of Statisticians (University of Helsinki,
July 24-28, 2017).}

\begin{abstract}
We develop a monitoring procedure to detect changes in a large approximate factor
model. Letting $r$ be the number of common factors, we base our statistics on the
fact that the $\left( r+1\right) $-th eigenvalue of the sample covariance matrix is
bounded under the null of no change, whereas it becomes spiked under changes.
Given that sample eigenvalues cannot be estimated consistently under the null, we
randomise the test statistic, obtaining a sequence of \textit{i.i.d} statistics, which are
used for the monitoring scheme. Numerical evidence shows a very small probability
of false detections, and tight detection times of change-points.
\end{abstract}

  \maketitle

\newpage

\section{Introduction\label{intro}}

In this paper, we investigate the issue of testing for the stability of a
large factor model: 
\begin{equation}
X_{i,t}=a_{i}^{\prime }f_{t}+u_{i,t},  \label{fmscalar}
\end{equation}%
where $\{X_{i,t},\ 1\leq i\le N,\ 1\le t\le T\}$ is a panel of $N$
time series observed for $T$ periods; $a_i$ and $f_{t}$ are latent
vectors of loadings and factors, respectively, both of dimension $r<N$ and
representing the ``signal'' component of the data, as opposed to the idiosyncratic ``noise'' $u_{i,t}$. In particular, we focus
on the \textit{sequential monitoring} of the stability of (\ref{fmscalar}) -
that is, we propose a test to check whether there are any breaks in (\ref%
{fmscalar}) as new data come in. Factor models have been paid significant
attention in virtually all applied sciences, as a tool to reduce
dimensionality while preserving the information content of a large dataset.
In particular, in the context of social sciences and economics, the use of
factor models has been popularised by the seminal paper by %
\citet{chamberlainrothschild83}; thereafter, factor models have acquired a
huge popularity in various applications, such as business cycle analysis,
asset pricing and economic monitoring and forecasting -- see the review by %
\citet{SW11} for a comprehensive list of references.

Model \eqref{fmscalar} is usually characterized by the identifying
assumption that, as $N\to\infty$, the covariance matrix of $%
\{X_{i,t}\}_{i=1}^N $ has $r$ spiked eigenvalues diverging to infinity,
while the remaining ones stay bounded for any $N$. Numerous contributions
have developed a full-fledged inferential theory for (\ref{fmscalar}) under general assumptions, such as weak serial and
cross-correlation of the error terms $u_{i,t}$. In particular, in the case
of stationary data the estimation by means of principal component analysis
of the ``signal'' part of (\ref{fmscalar}) has been developed for
high-dimensional, i.e. large $N$, data, e.g.~by \citet{bai03} and %
\citet{FLM13}. The literature has also produced many results on the
determination of the number of common factors $r $ -- see, \textit{inter alia}%
, \citet{baing02}, \citet{ABC10}, \citet{onatski10}, \citet{ahnhorenstein13}%
, and \citet{trapani17}. The factors in equation (\ref{fmscalar}) have also
proven to be very effective to forecast large datasets, overcoming the curse
of dimensionality issue -- see e.g. \citet{stockwatson02JASA}. Extensions to
the case in which the common factors $f_{t}$ are explicitly allowed to have
a linear process representation, have been studied also -- see e.g. %
\citet{FGLR09}. 

In comparison with this huge body of literature, the issue of testing for
the structural stability of (\ref{fmscalar}) can be still considered
underdeveloped, with some notable exceptions. Indeed, %
\citet{stockwatson02JASA} and \citet{bates2013} argue that, at least in the
presence of ``small'' breaks and a constant number of factors, inference on
the factor space is not hampered, thus making the change-point problem less
compelling than in other contexts. Nevertheless, stylised facts show that in
many applications the assumptions of a negligible break size and a stable
number of factors are not, in general, correct. Most importantly, it has
been argued that, in presence of a crisis, co-movements become stronger,
which may suggest that the economy is driven by a different number of factors than in quieter periods -- see e.g. \citet{SW09}, \citet{cheng2016} and \citet{li2017rank}.
In such cases, the impact of a
change-point is bound to invalidate standard inference and subsequent
applications such as forecasting. Recently, the literature has proposed a
series of tests for the in-sample detection of breaks in factor structures:
examples include the works by \citet{breitung2011}, \citet{chen2014}, %
\citet{han2014}, \citet{corradi2014}, \citet{yamamoto2015}, \citet{cheng2016}, \citet{baltagi2015}, \citet{massacci2017} and %
\citet{BCF16}. 

Sequential detection of breaks in (\ref{fmscalar}) is important for at least
four reasons. First, the general motivation put forward by \citet{CSW96}
holds true in the context of factor models also: it is important to verify
whether a model, which has been valid thus far, is still capable of
adequately approximate the behaviour of new data. Second, the aforementioned
(substantial) empirical evidence that factor structures do tend to change
over time, especially in presence of a crisis, illustrates the importance of a
timely detection of such changes. Third, inference on factor models can be
severely marred by the presence of a break (see the comments in %
\citet{baltagi2015}), which again shows the importance of detecting a
break in real time, rather than
realising this \textit{a posteriori} after inference has been carried out
and employed, e.g. for the purpose of forecasting. Finally, in the context of economics and finance, data are collected and
made available automatically, so that the cost of monitoring is almost negligible, especially if
compared with the potential costs of employing a model which is no longer
valid. Sequential detection of breaks in a univariate or small dimensional,
i.e.~finite $N$, setting has been studied e.g.~in \citet{lai1995}, %
\citet{CSW96}, \citet{aue2004delay}, \citet{lajos04}, \citet{AG06}, %
\citet{lajos07}, \citet{brodsky2010}, \citet{aue12}, \citet{KTK2015}, and %
\citet{GKS13}.

\subsection{Hypotheses of interest and main results of the paper}

There are several possible ways in which model (\ref{fmscalar}) may undergo
a change at a point in time $\tau $; however, despite such a wide variety,
in all cases it may be argued that a change in the factor structure of the
data will result in a change in the covariance matrix of $%
\{X_{i,t}\}_{i=1}^N $. More specifically, since common factors determine the
presence and number of spiked eigenvalue of the covariance of $%
\{X_{i,t}\}_{i=1}^N$ (defined as eigenvalues which are not bounded, but grow
with the dimension of the dataset), it is natural to investigate whether a
change has occurred in the factor structure of (\ref{fmscalar}) by verifying
whether changes have occurred in the spectrum of the covariance matrix.
Formally, in this paper we test for the null hypothesis that the factor
structure does not change, viz.:%
\begin{equation*}
H_{0}:X_{i,t}=\sum_{j=1}^r a_{ij}f_{j,t}+u_{i,t}\text{, \ \ }1\leq t\leq T.
\end{equation*}%
As far as alternatives are concerned, we focus on two different possible
breaks at a point in time $\tau $: {(1)} changes in the loadings attached to
one or more common factor:%
\begin{equation}
H_{A,1}:\left\{ 
\begin{array}{l}
X_{i,t}=\sum_{j=1}^r a_{ij}f_{j,t}+u_{i,t} \\ 
X_{i,t}=\sum_{j=1}^r \widetilde{a}_{ij}f_{j,t}+u_{i,t}%
\end{array}%
\text{ for }%
\begin{array}{l}
1\leq t< \tau \\ 
\tau \leq t\leq T%
\end{array}%
\right. ,  \label{h-a1}
\end{equation}%
where $\widetilde{a}_{ij}\neq a_{ij}$  for all $i$ and at least one value of $j$, 
and {(2)} the appearance of $q\ge 1$ new factors: 
\begin{equation}
H_{A,2}:\left\{ 
\begin{array}{l}
X_{i,t}=\sum_{j=1}^r a_{ij}f_{j,t}+u_{i,t} \\ 
X_{i,t}=\sum_{j=1}^r a_{ij}f_{j,t}+\sum_{j=1}^q b_{ij}g_{j,t}+u_{i,t}%
\end{array}%
\text{ for }%
\begin{array}{l}
1\leq t< \tau \\ 
\tau \leq t\leq T%
\end{array}%
\right. .  \label{h-a2}
\end{equation}

Hypothesis $H_{A,1}$ is the typical case considered in all the above cited
literature on change-points in factor models. A consequence of (\ref{h-a1}) is that, under the alternative, a model with $r$ common factors and changing loadings can be re-written as a model with a total number of factors ranging between $r+1$ and $2r$ common factors, defined as the original common factors multiplied by a pre- and post-break dummy variable. This key property is heavily exploited in the literature. On the other hand hypothesis $%
H_{A,2}$ has received less attention from the literature -- see for example %
\citet{cheng2016} and \citet{BCF16}. Whilst in this paper we mainly focus on 
$H_{A,1}$ and $H_{A,2}$, other alternatives, as disappearing factors or less pervasive changes in the loadings, can also be accommodated in our
framework -- see the discussion in Section \ref{discussion}.

We show that, under both $H_{A,1}$ and $H_{A,2}$, the $\left( r+1\right) $%
-th largest eigenvalue of the covariance matrix of $%
\{X_{i,t}\}_{i=1}^N$ becomes unbounded at time $\tau$, passing to infinity
as fast as the sample size $N$. Conversely, it stays bounded under the null
of no break. Thus, we base our test on the estimated $\left( r+1\right) $%
-th eigenvalue of the sample covariance matrix of $%
\{X_{i,t}\}_{i=1}^N$ computed using a rolling window. Although using the sample eigenvalues of the sample
covariance matrix for testing is not uncommon in the context of factor
models (\citet{onatski10,trapani17}), in our context such an approach is
fraught with difficulties. The main issue is that, under the null of no
break, the $\left( r+1\right) $-th sample eigenvalue does not have
a known distribution, and indeed it cannot even be estimated consistently:
as \citet{wang16} explain, there is too much noise (due to $N$ being large)
to be able to identify the small signal coming from a bounded eigenvalue.


Given that the only thing we know is that the $\left( r+1\right) $-th sample
eigenvalue may be bounded or unbounded, we propose to use a randomised test
in order to regularize the problem. Randomisation is a
widely employed approach, dating back at least to \citet{pearson50}; various
authors have employed different ways of introducing randomness
into a statistic -- see e.g. \citet{corradi2006}, \citet{bandi2014}, and %
\citet{trapani17}. Our methodology is based on the
same approach, but with a different scope. In essence,
the approach which we propose takes, at each point in time $t$, the $\left( r+1\right) $-th
sample eigenvalue as input, and returns, as output, an \textit{i.i.d.}
sequence, with known (asymptotic) distribution, first and second moments
that can be approximated with a negligible error, and finite moments up to
any order. Such sequence is then used to replace the $\left( r+1\right) $-th
sample eigenvalue in the construction of the monitoring process, thus
allowing us to use the standard asymptotic theory already developed for partial sum processes of \textit{i.i.d.} sequences -- see %
\citet{lajos04} and \citet{KTK2015}. Although our results are derived
conditionally on the sample (see the comments in Section \ref{monitoring} on
the meaning of randomisation under sample conditioning), we construct a
monitoring procedure which falsely identifies a break under the null with
probability smaller than a prescribed level, and which identifies a break
with probability one when this is present. This is a desirable feature of
sequential testing since as more data come in the probability of type I
errors is anyway likely to increase -- see for example the comments in
Chapter 9 by \citet{Sen}. Indeed, numerical evidence suggests that our
procedure works extremely well, with a short delay in finding breaks. 
In principle, our test can be applied also under more general circumstances,
including the presence of weak factors or less pervasive loadings changes, the case of heteroskedastic
idiosyncratic components, and the disappearance of one or more
factors. All these extensions are discussed in Sections \ref%
{discussion} and \ref{conclusions}. 

The rest of the paper is organised as follows. In Section \ref{assumptions}
we spell out the main assumptions, and we study the inference on the $\left(
r+1\right) $-th eigenvalue of the covariance matrix. Section \ref%
{monitoring} discusses the construction of the test statistic, including the
double randomisation procedure and all the relevant intermediate results.
Some straightforward extensions of our framework to more general
circumstances are discussed in Section \ref{discussion}. Numerical evidence
from Monte Carlo experiments and a real data application on US industrial production monthly data are given in
Sections \ref{numerics} and \ref{emp}, respectively. Section \ref{conclusions}\ discusses further
possible extensions and concludes. All proofs are in the Appendix.

NOTATION. We let $C_{0},C_{1},...$ denote generic, finite positive constants that do
not depend on the sample size, and whose value may change from line to line;
\textquotedblleft $\rightarrow $\textquotedblright\ denotes the ordinary
limit; orders of magnitude for an a.s. convergent sequence (say $s_{T}$) are
denoted as $O_{a.s.}\left( T^{\varsigma }\right) $ and $o_{a.s.}\left( T^{\varsigma
}\right) $ when, for some $\epsilon >0$ and $\tilde{T}<\infty $, $P\left[
\left\vert T^{-\varsigma }s_{T}\right\vert <\epsilon \text{ for all }T\geq 
\tilde{T}\right] =1$ and $T^{-\varsigma }s_{T}\rightarrow 0$ a.s.,
respectively; $I_{A}\left( x\right) $ is the indicator function of a set $A$. Finally, we assume without loss of generality that all random variables and processes are defined on a common probability space $(\Omega, \mathcal{F}, P)$ with outcomes $\omega \in \Omega$. 

\section{Assumptions and preliminary theory\label{assumptions}}

Consider the factor model in (\ref{fmscalar}), where we now make explicit the possibility
of changes over time in the ``signal'' component 
\begin{equation}  \label{fmscalar2}
X_{i,t}=a_{i}^{\prime }(t)f_{t}+u_{i,t}, \quad 1\le i\le N,\; 1\le t\le T.
\end{equation}%
We use $r\left( t\right) $\ to denote the number of factors at
a given time $t$, i.e. the vectors of loadings $a_{i}(t)$ and of factors $%
f_{t}$ have dimension $r(t)$. Consider also the matrix form of (\ref%
{fmscalar2}):%
\begin{equation}
X_{t}=A(t)f_{t}+u_{t}, \quad 1\le t\le T,  \label{fmvector}
\end{equation}%
where, $A(t)= \left[ a_{1}(t)|...|a_{N}(t)\right] ^{\prime }$ is the $N\times r(t)$
loadings matrix and $u_{t}= \left[ u_{1,t},...,u_{N,t}\right] ^{\prime }$ is
the idiosyncratic component. Under $H_{A,1}$ and $H_{A,2}$ (see also \eqref{h-a1} and \eqref{h-a2}), we define $\widetilde{a}_{i} =\left( \widetilde{a}_{i,1} ,...,\widetilde{a}%
_{i,r}\right)^{\prime }$ and $\widetilde{A}=%
\left[ \widetilde{a}_{1} |...|\widetilde{a}_{N} \right] ^{\prime }$, and $b_{i} =\left(
b_{i,1} ,...,b_{i,q}\right) ^{\prime }$ and $%
B =\left[ b_{1} |...|b_{N} %
\right] ^{\prime }$. Using this notation, we are interested in testing the null-hypothesis 
$$
H_{0}: A(t)=A \text{, \ \ }1\leq t\leq T,
$$
versus the alternatives
$$
H_{A,1}: \left\{ 
\begin{array}{l}
A(t)=A \\ 
A(t)=\widetilde A
\end{array}%
\text{ for }%
\begin{array}{l}
1\leq t< \tau \\ 
\tau \leq t\leq T%
\end{array}%
\right. ,
$$
and
$$
H_{A,2}:\left\{ 
\begin{array}{l}
A(t)=A \\ 
A(t)=[A \vert B]
\end{array}%
\text{ for }%
\begin{array}{l}
1\leq t< \tau \\ 
\tau \leq t\leq T%
\end{array}%
\right. 
$$

We define the covariance matrix of the data at time $t$ as $\Sigma
_{X}\left( t\right) =E\left( X_{t}X_{t}^{\prime }\right) $, assuming for simplicity, and without loss of generality, that $X_t$ has zero-mean. Consider the
(population) rolling covariance matrix 
\begin{equation}
\Sigma _{m}\left( t\right) =\frac{1}{m}\sum_{k=t-m+1}^{t}\Sigma _{X}\left(
k\right) ,\quad m\leq t\leq T,  \label{rolling-1}
\end{equation}%
and its sample counterpart 
\begin{equation}
\widehat{\Sigma }_{m}\left( t\right) =\frac{1}{m}%
\sum_{k=t-m+1}^{t}X_{k}X_{k}^{\prime },\quad m\leq t\leq T.
\label{rolling-2}
\end{equation}%
Based on (\ref{rolling-1}) and (\ref{rolling-2}), in what follows $m$ will denote our sample size when estimating the model; hence, our asymptotics
is for $m\rightarrow \infty $. We assume that for the first 
$m$ periods no change-point is present and we have $r(t)= r$ factors for all $t\le m$.  Moreover,
for simplicity, we also assume that our monitoring procedure
will last until $T>m$. Therefore, the total number of observations $T$
includes both the estimation and the monitoring period. Note that, in real
applications, the monitoring may be expected to go on indefinitely, so that $%
T\rightarrow \infty $.

We start with the following assumption.

\begin{assumption}
\label{ass-1}It holds that \textit{(i)} $E\left( X_{i,t}\right) =0$ for all $%
1\le i\le N$ and $1\le t\le T$; \textit{(ii)} $E(f_{j,t}u_{i,t})=0$ for all $%
i,j,t$; (iii) $r(t)=r$ for $1\leq t\le m$; (iv) $r(t)<N$ and finite for $1\leq t\leq T$
and for all $N\in \mathbb{N}$.
\end{assumption}

Parts \textit{(i)} and \textit{(ii)} of the assumption are made only for
convenience and could be relaxed. Clearly from part \textit{(iii)} we have
that, in presence of breaks, the change-point location $\tau$ is such that $%
\tau> m$. Finally, part \textit{(iv)} is a reasonable requirement for the
number of factors to be finite at any point in time. Note that, under $H_0$ and $H_{A,1}$ we have that $r(t)=r$ for all $m\le t\le T$, while under $H_{A,2}$ $r(t)=r$ for $1\le t< \tau$ and  $r(t)=(r+q)$ for $\tau\le t\le T$.

By Assumption \ref%
{ass-1} the covariance is decomposed as 
\begin{equation*}
\Sigma _{X}\left( t\right) =A\left( t\right) \Sigma _{F}\left( t\right)
A\left( t\right) ^{\prime }+\Sigma _{u}\left( t\right) ,
\end{equation*}%
having defined $\Sigma _{F}\left( t\right) = E\left( f_{t}f_{t}^{\prime
}\right) $ and $\Sigma _{u}\left( t\right) = E\left( u_{t}u_{t}^{\prime
}\right) $. %
Henceforth, we denote the $k$-th largest eigenvalue of $\Sigma
_{m}\left( t\right) $ as $\lambda ^{\left( k\right) }\left( t\right) $, the $%
k$-th eigenvalue of $A\left( t\right) \Sigma _{F}\left( t\right)
A\left( t\right) ^{\prime }$ as $\gamma ^{\left( k\right) }\left( t\right) $%
; and, finally, the $k$-th eigenvalue of $\Sigma _{u}\left(
t\right) $ as $\omega ^{\left( k\right) }\left( t\right) $; similarly, we
denote the $k$-th largest eigenvalue of $\widehat{\Sigma }%
_{m}\left( t\right) $ as $\widehat{\lambda }^{\left( k\right) }\left(
t\right) $. 


In order to derive our results on the population and sample
eigenvalues, we make the following assumptions.

\begin{assumption}
\label{ass-2}It holds that \textit{(i)} $\underline{C}_{k}(t)N\leq \gamma
^{\left( k\right) }\left( t\right) \leq \overline{C}_{k}\left( t\right) N$
for all $1\leq k\leq r\left( t\right) $, and $0<\underline{C}_{k}(t)\leq 
\overline{C}_{k}(t)<\infty $ and for $m\leq t\leq T$; 
\textit{(ii)} $\omega ^{\left(
k\right) }\left( t\right) \leq C_{0}$ for all $1\leq k\leq N$ and $m\leq
t\leq T$.
\end{assumption}

\begin{assumption}
\label{ass-3}It holds that \textit{(i)} $E\left\vert X_{i,t}\right\vert
^{4+\epsilon }\leq C_{0}$ for all $1\leq i\leq N$, $1\leq t\leq T$ and some $%
\epsilon >0$; \textit{(ii)}$\ E\left[ \max_{t_{0}\leq \tilde{t}\leq
t_{0}+m-1}\left\vert \sum_{t=t_{0}}^{\tilde{t}}X_{h,t}X_{j,t}-E\left( X_{h,t}X_{j,t}\right) \right\vert ^{2}\right] $ $\leq 
$ $C_{1}m$ for all $1\leq h,j\leq N$ and $1\leq t_{0}\leq T-m+1$.
\end{assumption}

Assumption \ref{ass-2} is typical of high-dimensional factor analysis and is
analogous to the assumptions in \citet{chamberlainrothschild83} and %
\citet{FGLR09}. In particular, as far as the non-zero $\gamma ^{\left(
k\right) }\left( t\right) $'s are concerned, part \textit{(i)} of the
assumption requires that they diverge to positive infinity, as $N\rightarrow
\infty $, at a rate $N$. 

Equivalently, we could follow \citet{baing02} and \citet{FLM13} and require
the more primitive assumptions that $\Sigma _{F}\left( t\right) $ is
positive definite (which entails that common factors are identified), and
that $N^{-1}A\left( t\right) ^{\prime }A\left( t\right) $ tends to a
positive definite matrix. This is tantamount to assuming that $\gamma
^{\left( k\right) }\left( t\right) $ passes to infinity at a rate $N$.
Indeed, consider - for the sake of the notation - the case of constant
loadings, viz. $A\left( t\right) =A$, and constant covariance matrix for the
common factors, viz. $\Sigma _{F}\left( t\right) =\Sigma _{F}$; then, using
Theorem 7 in \citet{merikoski2004}%
\begin{equation*}
N\nu ^{\left( \min \right) }\left( \Sigma _{F}\right) \nu ^{\left( k\right)
}\left( \frac{A^{\prime }A}{N}\right) \leq \nu ^{\left( k\right) }\left( A\Sigma _{F}A^{\prime }\right)
\leq N\nu ^{\left( \max \right) }\left( \Sigma _{F}\right) \nu ^{\left(
k\right) }\left( \frac{A^{\prime }A}{N}\right) ,
\end{equation*}%


where $\nu ^{\left( k\right) }\left( \cdot \right) $ denotes the $k$-th
largest eigenvalue of a matrix. 
When following the same reasoning in the presence of change-points, the above result provides a link between  $\gamma ^{\left( k\right)
}\left( t\right) $ and  the $k$-th largest eigenvalue of $N^{-1}A^{\prime }(t)A(t)%
$. Note that it is also possible to assume that $\gamma ^{\left(
k\right) }\left( t\right) \rightarrow \infty $ as $N\rightarrow \infty $ at
a slower rate than $N$, which is known as having \textquotedblleft weak
factors\textquotedblright; we discuss this case
in Section \ref{weak}. 

As far as the $\omega ^{\left( k\right) }\left( t\right) $'s are concerned,
in part \textit{(ii)} of the assumption, the same condition could be derived
from the assumptions in \citet{FLM13} -- see also \citet{baing02}. Note also that we do not require the $%
\omega ^{\left( k\right) }\left( t\right) $'s to be constant over $t$:
unconditional heteroskedasticity is allowed for, in principle -- see also the
comments in Section \ref{discussion}. Assumption 2 determines the behaviour
of the population eigenvalues of $\Sigma _{m}\left( t\right) $. In
particular, at $t=m$, by Weyl's inequality we have that $\lambda ^{\left(
k\right) }\left( m\right) \geq C_{k}(m)N$ for $1\leq k\leq r$, while $\lambda
^{\left( k\right) }\left( m\right) \leq C_{0}$ for $r+1\leq k\leq N$. This
condition implies the existence of an eigen-gap which allows us to identify $r$ in the pre-break sample.


As far as Assumption \ref{ass-3} is concerned, part \textit{(ii)} is a high-level condition which, in
essence, poses a constraint on the amount of serial correlation that one can
have in the process $\{X_{h,t}X_{j,t}\}_{t=1}^T$ and therefore, albeit
indirectly, in $\{X_{i,t}\}_{t=1}^T$. In general, this assumption is
satisfied by any linear process with summable fourth cumulants (see e.g., %
\citet{hannan1970}, Theorem 6, page 210). Some examples under which
Assumption \ref{ass-3} holds are reported in \citet{trapani17} and include
the case of stationary, causal processes -- see \citet{wu05}. This family of
processes in turn includes several popular examples such as Volterra series
and ARCH/GARCH processes, thus allowing for the case of conditional
heteroskedasticity. 

Finally, note that Assumptions \ref{ass-2} and \ref{ass-3} allow for some
degree of cross-sectional and serial dependence in the panel of
idiosyncratic components, $\left\{ u_{i,t}\text{, }1\leq i\leq N\text{, }1\leq t\leq
T\right\} $; thus, (\ref{fmscalar}) defines an \textquotedblleft
approximate\textquotedblright\ factor model, as opposed to an
\textquotedblleft exact\textquotedblright\ one, which would require
cross-sectionally and serially \textit{i.i.d.} errors.

The following result characterizes the behaviour of the $(r+1)$-th 
eigenvalue of $\Sigma _{m}(t)$.

\begin{lemma}
\label{rolling-eigenvalues}Under Assumptions \ref{ass-1} and \ref{ass-2}, it
holds that 
\begin{equation}
\lambda ^{\left( r+1\right) }\left( t\right) \leq C_{0}\text{, \ \ }m\leq
t\leq T,\text{ under }H_{0}.  \label{lambda-null}
\end{equation}%
Further, it holds that 
\begin{eqnarray}
&&\lambda ^{\left( r+1\right) }\left( t\right) \left\{ 
\begin{array}{ll}
\leq C_{0} & m\le t<\tau, \\ 
\geq C_{1}\min \{ \frac{t-\tau +1}{m},\frac{\tau+m-t-1}{m}\} N & \tau \leq t< \tau +m-1, \\ 
\leq C_{0} & \tau +m-1\le t\le T,\text{ \ \ under }H_{A,1},%
\end{array}%
\right.   \label{lambda-population} \\
&&\lambda ^{\left( r+1\right) }\left( t\right) \left\{ 
\begin{array}{ll}
\leq C_{0} & m\le t<\tau, \\ 
\geq C_{1}\frac{t-\tau +1}{m}N & \tau \le t< \tau +m-1, \\ 
\geq C_{1}N & \tau +m-1\le t \le T,\text{ \ \ under }H_{A,2}.
\end{array}%
\right.   \label{lambda-population-2}
\end{eqnarray}%
%
%
%
%
%
\end{lemma}


 The sample counterpart to Lemma \ref%
{rolling-eigenvalues} is the following result, derived in \citet{trapani17}.

\begin{lemma}
\label{rolling-eigenvalues-2}Under Assumptions \ref{ass-1} and \ref{ass-3},
it holds that%
\begin{equation}
\widehat{\lambda }^{\left( r+1\right) }\left( t\right) =\lambda ^{\left(
r+1\right) }\left( t\right) +O_{a.s.}\left( \frac{N}{m^{1/2}}l\left(
m,N\right) \right) , \quad m\leq t\leq T,  \label{lambda-estimate}
\end{equation}%
where \begin{equation*}
l\left( m,N\right) =\left( \ln N\right)^{1+\epsilon } \left( \ln m\right)^{%
\frac{1+\epsilon }{2}},
\end{equation*}%
 for any $\epsilon >0$.
%
%
%
%
%
\end{lemma}

Lemma \ref{rolling-eigenvalues-2} provides a strong rate for the estimation
error $(\widehat{\lambda }^{\left( r+1\right) }\left( t\right) -\lambda
^{\left( r+1\right) }\left( t\right) ) $, which is valid for any combination
of $N$ and $m$, and indeed for all estimated eigenvalues, $\widehat{\lambda }%
^{\left( k\right) }\left( t\right) $ for $1\leq k\leq N$. The lemma does not
require any assumption on $\lambda ^{\left( k\right) }\left( t\right) $:
some of these may be non-distinct, non well-separated, or even equal to
zero. Equation (\ref{lambda-estimate}) states that the estimation error can
be quite large. It is, however, comparatively small for the spiked
eigenvalues, which, by Assumption 2, are of order $N$. Conversely, the error term in (\ref{lambda-estimate})
can be quite large for the bounded eigenvalues; in this case, the rate is
probably not the sharpest one, although it suffices for the construction of
the monitoring procedure. The result of Lemma \ref{rolling-eigenvalues-2}
can also be compared with the results from Random Matrix literature for
spiked covariance models where however $\lambda^{(k)}$ is finite for all $1\le k\le N$ and $N\in\mathbb N$ -- see e.g. \citet{karoui2007}, \citet{paul2007}, %
\citet{JL2009}, \citet{JM09}, \citet{benaych2011,benaych2012}, \citet{BY12}, and \citet{OMH14}.

\section{Testing procedure and asymptotics\label{monitoring}}

In this section, we propose an algorithm to \textquotedblleft
regularise\textquotedblright\ the behaviour of the eigenvalues so as to be
able to construct a monitoring procedure. As a consequence of Lemmas \ref{rolling-eigenvalues} and \ref{rolling-eigenvalues-2},
we are unable to use $\widehat{\lambda }^{\left(
r+1\right) }\left( t\right)$, due to the lack of a known limiting distribution
under the null, and of the dependence structure across $t$. We therefore propose a
randomisation algorithm, whose output is a sequence of \textit{i.i.d.}
random variables with finite moments of arbitrarily high order and, under
the null, (asymptotically) chi-square distributed. We subsequently employ
(the standardised version of) such random variables to construct a partial
sum process, which we use as the relevant test statistic in an analogous way
as \citet{lajos04} and \citet{lajos07}.

\subsection{The randomisation algorithm\label{doublerandom}}

Define $\delta \in \left( 0,1\right) $ such that%
\begin{equation}
\delta \left\{ 
\begin{array}{c}
>0\text{ \ \ \ \ \ \ \ \ \ } \\ 
>1-\frac{1}{2}\frac{\ln m}{\ln N}%
\end{array}%
\right. \text{ according as }\left. 
\begin{array}{l}
N\leq m^{1/2} \\ 
N>m^{1/2}%
\end{array}%
\right. ;  \label{delta}
\end{equation}%
note that the choice of $\delta $ is uniquely determined by $N$ and $m$, with no
need to estimate it. We consider the statistic%
\begin{equation}
\phi _{N,m}\left( t\right) =g\left( \frac{N^{-\delta }\widehat{\lambda }%
^{\left( r+1\right) }\left( t\right) }{\frac{1}{N}\sum_{k=1}^{N}\widehat{%
\lambda }^{\left( k\right) }\left( t\right) }\right) ,\quad m\leq t\leq T,
\label{phi}
\end{equation}%
where $g\left( \cdot \right) $ is a monotonically increasing function such
that $g\left( 0\right) =0$ and $\lim_{x\rightarrow \infty }g\left( x\right)
=\infty $; in this paper, we use $g(a)=a$, but other choices are also
possible. The denominator in (\ref{phi}) makes the argument of $g\left(
\cdot \right) $ scale invariant.

The quantity $\delta $, defined in (\ref{delta}), plays a very important
role in the remainder of the paper. Based on Lemma \ref%
{rolling-eigenvalues-2}, it can be expected that $\widehat{\lambda }^{\left(
r+1\right) }\left( t\right) $ may diverge to positive infinity even when $%
\lambda ^{\left( r+1\right) }\left( t\right) $ is bounded; in this case, the
divergence rate is $O\left( Nm^{-1/2}\right) $, modulo the logarithmic
terms. On the other hand, $\widehat{\lambda }^{\left( r+1\right) }\left(
t\right) $ diverges at the faster rate $O\left( N\right) $ under the
alternative. The purpose of $\delta $ is to annihilate the estimation error:
based on (\ref{delta}), it can be seen that $N^{\delta }$ is larger than $%
Nm^{-1/2}l\left( m,N\right) $: thus, under the null of no break, it
can be expected that $N^{-\delta }\widehat{\lambda }^{\left( r+1\right)
}\left( t\right) $ will drift to zero. Under the alternative, it still
passes to infinity (since $\delta <1$), albeit at a slower rate than $%
\lambda ^{\left( r+1\right) }\left( t\right) $ itself. Note that this would
hold also for very large values of $N$: indeed, no restriction is required
between the relative rate of divergence of $N$ and $m$ as they pass to
infinity, and one could also allow for $N=\exp \left( m\right) $; in this
case, after some algebra it can be shown that $\delta \in \left( 1-\frac{1}{2%
}\frac{\ln m}{m},1\right) $, which still yields that $N^{-\delta }\widehat{%
\lambda }^{\left( r+1\right) }\left( t\right) $ drifts to zero or diverges
to infinity according as the null or the alternative is true.

On account of the comments above, and of Lemmas \ref{rolling-eigenvalues}
and \ref{rolling-eigenvalues-2}, it holds that 
\begin{align*}
& \lim_{N,m\rightarrow \infty }\phi _{N,m}\left( t\right) =g\left( 0\right)
=0,\;\text{ w.p. 1, when }\;N^{-\delta }{\lambda }^{\left(
r+1\right) }\left( t\right) \rightarrow 0, \\
& \lim_{N,m\rightarrow \infty }\phi _{N,m}\left( t\right) =g\left( \infty
\right) =\infty ,\;\text{ w.p. 1, when }\;N^{-\delta }{\lambda }%
^{\left( r+1\right) }\left( t\right) \rightarrow \infty .
\end{align*}%
We therefore have that 
\begin{equation}
\lim_{N,m\rightarrow \infty }\phi _{N,m}\left( t\right) =0,\quad m\leq t\leq
T,\text{ \ \ under }H_{0}.  \notag
\end{equation}%
Henceforth, we define $t_{N,m}^{\ast }$ as the point in time such that $t_{N,m}^{\ast }\ge \tau$ and
\begin{equation}
\lim_{N,m\rightarrow \infty }\frac{N^{1-\delta }}{m}\left( t_{N,m}^{\ast
}-\tau +1\right) =\infty .\label{delaytstar}
\end{equation}%
Similarly, we define the point in time $t_{N,m}^{\ast\ast } \leq \tau+m-1$ such that
\begin{equation}
\lim_{N,m\rightarrow \infty }\phi _{N,m}\left( t\right) =\left\{ 
\begin{array}{ll}
0 & m\leq t<\tau , \\ 
\infty & t_{N,m}^{\ast }\leq t<t_{N,m}^{\ast\ast }, \\ 
0 & \tau +m-1\leq t\leq T,%
\end{array}%
\right. ,\text{ \ \ under }H_{A,1}.  \notag
\end{equation}%
Clearly 
\begin{equation}
\lim_{N,m\rightarrow \infty }\phi _{N,m}\left( t\right) =\left\{ 
\begin{array}{ll}
0 & m\leq t<\tau , \\ 
\infty & t_{N,m}^{\ast }\leq t<\tau +m-1,%
\end{array}%
\right. ,\text{ \ \ under }H_{A,2};  \notag
\end{equation}%
Under $H_{A,1}$ for $\tau\le t<t_{N,m}^{\ast }$, $\phi _{N,m}\left( t\right) $ is growing from $0$ to $\infty $, and viceversa for $t_{N,m}^{\ast\ast }\le t < \tau+m-1$, while under $H_{A,2}$ for $\tau\le t<t_{N,m}^{\ast }$, $\phi _{N,m}\left( t\right) $
is growing from $0$ to $\infty $. Therefore, $t_{N,m}^{\ast }$ represents the first point in time in which we can hope to detect the change, hence is a lower bound for the delay in detection, while under $H_{A,1}$ $(t_{N,m}^{\ast\ast }-1)$ represents the last point in time in which we can hope to detect the change. In light of \eqref{delaytstar} and the results that follow, we show in Section \ref{delay} below that $t_{N,m}^{\ast }$ is at least of order $ m^{1/2}$ regardless of the values of $m$ and $N$.

Given that the results above entail that we only have rates for $\phi
_{N,m}\left( t\right) $, we propose a to use a randomised version of it,
built according to the following steps.

\begin{description}
\item \textit{\ {Step A1.}} At each given $t\ge m$, generate an \textit{%
i.i.d.} sample $\big\{
\xi _{j}(t)\big\} _{j=1}^{R}$ with common distribution $G_{\phi }$ such that 
$G_{\phi }(0) \neq 0$ or $1$.

\item \textit{\ {Step A2.}} For any $u$ drawn from a distribution $F_{\phi
}\left( u\right) $, define%
\begin{equation*}
\zeta _{j}\left( u;t\right) =I\left[ \xi _{j}(t)\leq u\phi _{N,m}^{-1}\left(
t\right) \right] .
\end{equation*}

\item \textit{\ {Step A3.}} Compute%
\begin{equation*}
\vartheta \left( u;t\right) =\frac{1}{\sqrt R}\sum_{j=1}^{R}\frac{ \zeta
_{j}\left( u;t\right) -G_{\phi }\left( 0\right) }{\sqrt{G_{\phi }\left(
0\right) \left[ 1-G_{\phi }\left( 0\right) \right] }} .
\end{equation*}

\item \textit{\ {Step A4.}} Compute%
\begin{equation*}
\Theta _{t}=\int_{-\infty }^{+\infty }\left\vert \vartheta \left( u;t\right)
\right\vert ^{2}dF_{\phi }\left( u\right) .
\end{equation*}
\end{description}

Although the details of the behaviour of $\Theta _{t}$ under the null and
the alternative are spelt out later on, a heuristic preview of the main
argument may be helpful. In essence, under the alternative the Bernoulli
random variable $\zeta _{j}\left( u;t\right) $ should be equal to $1$ or $0$
with probability $G_{\phi }\left( 0\right) $ and $1-G_{\phi }\left( 0\right) 
$ respectively, and thus have mean $G_{\phi }\left( 0\right) $. In this
case, when constructing $\vartheta \left( u;t\right) $, a Central Limit
Theorem holds and therefore we expect $\Theta _{t}$ to have a chi-square
distribution. On the other hand, under the null $\zeta _{j}\left( u;t\right) 
$ should be (heuristically) $0$ or $1$ with probability $0$ or $1$
(depending on the sign of $u$) - thus, its mean should be different than $%
G_{\phi }\left( 0\right) $ (and equal to $0$ or $1$ depending on the sign of 
$u$) and a Law of Large Numbers should hold. Note that, by construction,
conditionally on the sample the sequence $\{\Theta_t\}_{t=m}^T$ is
independent across $t$. In order to study $\Theta _{t}$, we need the
following assumptions.

\begin{assumption}
\label{ass-4}It holds that: \textit{(i)} $G_{\phi }\left( \cdot \right) $
has a bounded density; \textit{(ii)} $\int_{-\infty }^{+\infty
}u^{2}dF_{\phi }\left( u\right) <\infty $; (iii) $F_{\phi }\left(0 \right)<1 $.
\end{assumption}

\begin{assumption}
\label{restriction-1}
It holds that, as $\min \left( N,m,R\right) \rightarrow
\infty $: 
\begin{equation*}
\text{\it{(i) }}\ R^{1/2}\left[ g\left( N^{1-\delta }\frac{t-\tau +1}{m}\right) \right]
^{-1}\rightarrow 0,%
\begin{array}{ll}
\text{under }H_{A,1}, & \text{for }\ t_{N,m}^{\ast }\leq t<t_{N,m}^{\ast\ast }, \\ 
\text{under }H_{A,2}, & \text{for }\ t_{N,m}^{\ast }\leq t\leq T;%
\end{array}%
\end{equation*}%
(ii) $R^{1/2}\left[ g\left( N^{1-\delta }\right) \right] ^{-1}\rightarrow 0$
under $H_{A,1}$, for $\,\tau +m-1\leq t\leq T$.
\end{assumption}

Considering Assumption \ref{ass-4}, $G_{\phi }$ can be chosen as the
standard normal distribution, and $F_{\phi }$ as a discrete uniform
distribution. Assumption \ref{restriction-1}
provides a selection rule for $R$.

\bigskip
Let now $P^{\ast }$ represent the conditional probability with respect to $%
\{X_{i,t},1\leq i\leq N$, $1\leq t\leq T\}$; \textquotedblleft $\overset{{%
\mathcal{D}}^{\ast }}{\rightarrow }$\textquotedblright\ and
\textquotedblleft $\overset{P^{\ast }}{\rightarrow }$\textquotedblright\
denote, respectively, conditional convergence in distribution and in
probability according to $P^{\ast }$.

\begin{theorem}
\label{theta}Under Assumptions \ref{ass-1}-\ref{restriction-1}, as $\min
\left( N,m,R\right) \rightarrow \infty $, it holds that 
\begin{equation}
\Theta _{t}\overset{\mathcal{D}^{\ast }}{\rightarrow }\chi _{1}^{2},\ 
\begin{array}{ll}
\text{under }H_{A,1}, & \text{\ \ for }\ t_{N,m}^{\ast }\leq t<t_{N,m}^{\ast\ast }, \\ 
\text{under }H_{A,2}, & \text{\ \ for }\ t_{N,m}^{\ast }\leq t\leq T,%
\end{array}
\label{theta-null}
\end{equation}%
for almost all realisations of $\left\{ X_{i,t}\text{, }1\leq i\leq N\text{, 
}1\leq t\leq T\right\} $.

\noindent Under Assumptions \ref{ass-1}-\ref{ass-4}, as $\min \left(
N,m,R\right) \rightarrow \infty $, it holds that 
\begin{equation}
\frac{1}{R}\Theta _{t}\overset{P^{\ast }}{\rightarrow }\frac{\int_{-\infty
}^{+\infty }\left\vert I_{\left[ 0,\infty \right) }\left( u\right) -G_{\phi
}\left( 0\right) \right\vert ^{2}dF_{\phi }\left( u\right) }{G_{\phi }\left(
0\right) \left[ 1-G_{\phi }\left( 0\right) \right] },\ 
\begin{array}{ll}
\text{under }H_{0}, & \text{for }\ m\leq t\leq T, \\ 
\text{under }H_{A,1}, & \text{for }\ m\leq t<\tau, \\ 
& \text{and } \tau+m-1\leq t\leq T, \\ 
\text{under }H_{A,2}, & \text{for }\ m\leq t<\tau ,%
\end{array}
\label{theta-alternative}
\end{equation}%
for almost all realisations of $\left\{ X_{i,t}\text{, }1\leq i\leq N\text{, 
}1\leq t\leq T\right\} $.
\end{theorem}

Theorem \ref{theta} is an intermediate result: in order to
be able to construct a test for the \textquotedblleft
classical\textquotedblright\ null of no changes in the factor structure, it
is necessary to have a statistic which diverges under the null and is
bounded under the alternative. In particular, the behaviour under the null is - clearly - very important to ensure size control of the monitoring procedure. As can be noted, the reason why $\Theta
_{t} $ is bounded under the null is because we have constructed a statistic based on randomising the estimated eigenvalue $\widehat{\lambda}^{(r+1)}\left( t \right)$. Thus, it can be envisaged that randomising its reciprocal would yield the desired behaviour. Whilst this is theoretically possible, we recommend against it: as Lemma \ref{rolling-eigenvalues-2} shows, in this case, under the null, the behaviour of $\Theta
_{t} $ would be driven by a term proportional to (the inverse of) $N^{-\delta}\frac{N}{m^{1/2}}$: but since this estimate is only an upper bound, and thus not sharp (contrary to the case of randomising $\widehat{\lambda}^{(r+1)}\left( t \right)$ directly), it is unclear what the rate of divergence would be in this case.  \\
We therefore
propose to randomise $\Theta
_{t} $, with a second randomisation based on 
\begin{equation}
\psi _{N,m,R}\left( t\right) =h\left( \frac{\Theta _{t}}{\widetilde{l}\left(
N,m,R\right) }\right) , \quad m\le t\le T,  \label{psi-1}
\end{equation}%
where 
\begin{equation*}
\widetilde{l}\left( N,m,R\right) =\left( \ln N\right) ^{2+\epsilon }\left(
\ln m\right) ^{2+\epsilon }\left( \ln R\right) ^{2+\epsilon },
\end{equation*}%
for some $\epsilon >0$ - in practice, any small value of $\epsilon $ works
well.

In (\ref{psi-1}), the function $h\left( \cdot \right) $, similarly to $%
g\left( \cdot \right) $ in (\ref{phi}), is a monotonically increasing
function such that $h\left( 0\right) =0$ and $\lim_{x\rightarrow \infty
}h\left( x\right) =\infty $; again, we use $h\left( a\right) =a$.

Similarly to the case of $\phi _{N,m}\left( t\right) $, Theorem \ref{theta}
entails that
\begin{equation*}
\lim_{N,m,R\rightarrow \infty }\psi _{N,m,R}\left( t\right) =\infty ,\quad
m\leq t\leq T,\text{ \ \ under }H_{0},
\end{equation*}%
and  
\begin{equation*}
\lim_{N,m,R\rightarrow \infty }\psi _{N,m,R}\left( t\right) =\left\{ 
\begin{array}{ll}
\infty & m\leq t<\tau , \\ 
0 & t_{N,m}^{\ast }\leq t<t_{N,m}^{\ast\ast }, \\ 
\infty & \tau +m-1\leq t\leq T,%
\end{array}%
\right. ,\text{ \ \ under }H_{A,1},
\end{equation*}%
while 
\begin{equation*}
\lim_{N,m,R\rightarrow \infty }\psi _{N,m,R}\left( t\right) =\left\{ 
\begin{array}{ll}
\infty & m\leq t<\tau , \\ 
0 & t_{N,m}^{\ast }\leq t<\tau +m-1,%
\end{array}%
\right. ,\text{ \ \ under }H_{A,2}.
\end{equation*}%
Consider now the second randomisation.

\begin{description}
\item \textit{\ {Step B1.}} At each given $t\ge m$, generate an \textit{%
i.i.d.} sample $\big\{ 
\widetilde{\xi }_{j}(t)\big\} _{j=1}^{W}$ with common distribution $G_{\psi
} $ such that $G_{\psi}(0) \neq 0$ or $1$.

\item \textit{\ {Step B2.}} For any $u$ drawn from a distribution $F_{\psi
}\left( u\right) $, define%
\begin{equation*}
\widetilde{\zeta }_{j}\left( u;t\right) =I\left[ \widetilde{\xi }_{j}(t)\leq
u\psi _{N,m,R}^{-1}\left( t\right) \right] .
\end{equation*}

\item \textit{\ {Step B3.}} Compute 
\begin{equation*}
\gamma \left( u;t\right) =\frac{1}{\sqrt W}\sum_{j=1}^{W} \frac{\widetilde{%
\zeta _{j}}\left( u;t\right) -G_{\psi }\left( 0\right) } {\sqrt{ G_{\psi
}\left( 0\right) \left[ 1-G_{\psi }\left( 0\right) \right] }} .
\end{equation*}

\item \textit{\ {Step B4.}} Compute%
\begin{equation*}
\Gamma _{t}=\int_{-\infty }^{+\infty }\left\vert \gamma \left( u;t\right)
\right\vert ^{2}dF_{\psi }\left( u\right) .
\end{equation*}
\end{description}

The following assumptions are needed in order to study the asymptotic
behavior of $\Gamma _{t}$; note their similarity with Assumptions \ref{ass-4}
and \ref{restriction-1}.

\begin{assumption}
\label{ass-5}It holds that: \textit{(i)} $G_{\psi }\left( \cdot \right) $
has a bounded density; \textit{(ii)} $\int_{-\infty }^{+\infty
}u^{4}dF_{\psi }\left( u\right) <\infty $; (iii) $F_{\psi }\left( 0 \right)<1 $.
\end{assumption}

\begin{assumption}
\label{restriction-2}It holds that, as $\min \left( N,m,R,W\right)
\rightarrow \infty $ 
\begin{equation*}
W^{1/2}\left[h\left( \frac{R}{\widetilde{l}\left( N,m,R\right) }\right)%
\right]^{-1} \rightarrow 0.
\end{equation*}
\end{assumption}

As above, in Assumption \ref{ass-5} we can choose $G_{\psi }$ to be the
standard normal distribution, and $F_{\psi }$ to be a discrete uniform
distribution. The restrictions in Assumption \ref{restriction-2} provide a
selection rule for $W$. 

Let $P^{\dag }$ represent the conditional probability with respect to $%
\{X_{i,t},1\leq i\leq N$, $1\leq t\leq T\}$ and $\left\{ \xi _{j}(t),1\leq
j\leq R,\ m\leq t\leq T\right\} $; we use the notation \textquotedblleft $%
\overset{{\mathcal{D}}^{\dag }}{\rightarrow }$\textquotedblright\ and
\textquotedblleft $\overset{P^{\dag }}{\rightarrow }$\textquotedblright\ to
define, respectively, conditional convergence in distribution and in
probability according to $P^{\dag }$.

\begin{theorem}
\label{gamma}Under Assumptions \ref{ass-1}-\ref{restriction-2}, as $\min
\left( N,m,R,W\right) \rightarrow \infty $, it holds that%
\begin{equation}
\Gamma _{t}\overset{\mathcal{D}^{\dag }}{\rightarrow }\chi _{1}^{2},\ 
\begin{array}{ll}
\text{under }H_{0}, & \text{for }\ m\leq t\leq T, \\ 
\text{under }H_{A,1}, & \text{for }\ m\leq t<\tau \text{ and }\tau +m-1\leq
t\leq T, \\ 
\text{under }H_{A,2}, & \text{for }\ m\leq t<\tau ,%
\end{array}
\label{gamma-null}
\end{equation}%
for almost all realisations of $\{ X_{i,t}\text{, }1\leq i\leq N\text{, 
}1\leq t\leq T\} $ and $\{ \xi _{j}(t),1\leq j\leq R\text{, }m\leq
t\leq T\} $.

\noindent Under Assumptions \ref{ass-1}-\ref{restriction-1}, as $\min \left(
N,m,R,W\right) \rightarrow \infty $, it holds that 
\begin{equation}
\frac{1}{W}\Gamma _{t}\overset{P^{\dag }}{\rightarrow }\frac{\int_{-\infty
}^{+\infty }\left\vert I_{\left[ 0,\infty \right) }\left( u\right) -G_{\psi
}\left( 0\right) \right\vert ^{2}dF_{\psi }\left( u\right) }{G_{\psi }\left(
0\right) \left[ 1-G_{\psi }\left( 0\right) \right] },\ 
\begin{array}{ll}
\text{under }H_{A,1}, & \text{for }\ t_{N,m}^{\ast }\leq t<t_{N,m}^{\ast\ast }, \\ 
\text{under }H_{A,2}, & \text{for }\ t_{N,m}^{\ast }\leq t\leq T,%
\end{array}
\label{gamma-alternative}
\end{equation}%
for almost all realisations of $\{ X_{i,t}\text{, }1\leq i\leq N\text{, 
}1\leq t\leq T\} $ and $\{\xi _{j}(t)\text{, }1\leq j\leq R\text{, }%
m\leq t\leq T\}$.
\end{theorem}

Theorem \ref{gamma} is, again, an intermediate result. It states that $%
\Gamma _{t}$ has (asymptotically) a chi-square distribution under the null
of no breaks; further, by construction the sequence $\big\{\Gamma _{t}\big\}%
_{t=m}^T$ is independent across $t$ conditional on the sample. We now
discuss how these two basic facts can be employed in order to propose a
monitoring scheme for the on-line detection of breaks in the factor
structure.

\subsection{Sequential monitoring of factor models\label{sequential}}

We base our sequential monitoring procedure on the theory developed in %
\citet{lajos04}. Recall that, after collecting $m$ observations, we monitor
our model over the period $m+1\le t\le T$, which has size denoted as $T_m =
T-m$. We then consider a monitoring procedure based on the detector 
\begin{equation}  \label{eq:detector}
d\left( k;m\right) =\left\vert \sum_{t=m+1}^{m+k}\frac{\Gamma _{t}-1}{\sqrt{2%
}}\right\vert, \quad 1\le k\le T_m,
\end{equation}
which covers the entire monitoring period. In other words our detector is
made of the cumulative sum of the centered and standardized version of the
sequence $\{\Gamma_t\}_{t=m}^T$, obtained by double randomisation. 
Other detectors, differing form \eqref{eq:detector} only with respect to the start of the monitoring period, could be also suggested.
In particular, \citet{kirch17} suggest to use a rolling window, thus starting the monitoring procedure at $t=m+k-h+1$ for some $\underline h<h<k$, with $\underline h$ large enough. The asymptotic
properties of such alternative detector can be derived in a way
similar to the results proved in this section and therefore are not
discussed in this paper.
In light of Theorem \ref{gamma}, a break implies a shift in the mean of $\Gamma_t$ and therefore in the detector \eqref{eq:detector}. Therefore, our monitoring scheme looks for large deviations of $d(k;m)$ from its null-distribution.

Given the stopping rule 
\begin{equation}
\widehat{k}_{m}=\left\{ 
\begin{array}{l}
\inf \left\{ 1\leq k\leq T_{m}\text{, such that }d\left( k;m\right) \geq \nu
\left( k;m\right) \right\} , \\ 
T_{m}\text{ if the above does not hold in } 1\leq k\leq T_{m},%
\end{array}%
\right.  \label{tau-m}
\end{equation}%
we define the estimated change-point location as $\widehat{\tau }_{m}=%
\widehat{k}_{m}+m$. The threshold function in \eqref{tau-m} is defined as
(see \citet{lajos04} and \citet{lajos07})%
\begin{eqnarray}
\nu \left( k;m\right) &=&c_{\alpha ,m}\nu ^{\ast }\left( k;m\right) ,
\label{threshold-1} \\
\nu ^{\ast }\left( k;m\right) &=&m^{1/2}\left( 1+\frac{k}{m}\right) \left( 
\frac{k}{k+m}\right) ^{\eta },\text{ }\eta \in \left[ 0,\frac{1}{2}\right] ,
\label{threshold-2}
\end{eqnarray}%
where $c_{\alpha ,m}$ is a critical value corresponding to a pre-specified
level $\alpha $. Depending on the choice of $\eta $, the critical value is
defined as 
\begin{equation}
P\left( \sup_{0\leq t\leq 1}\frac{\left\vert B\left( t\right) \right\vert }{%
t^{\eta }}\leq c_{\alpha ,m}\right) =1-\alpha ,\;\text{ for }\eta \in \left[
0,\frac{1}{2}\right) ,  \label{soglia1}
\end{equation}%
where $\{B\left( t\right) ,\ 0\leq t\leq 1\}$ denotes a standard Wiener
process, or 
\begin{equation}
c_{\alpha ,m}=\frac{D_{m}-\ln \left[ -\ln \left( 1-\alpha \right) \right] }{%
A_{m}},\;\text{ for }\eta =\frac{1}{2},  \label{soglia2}
\end{equation}%
with $A_{m}=\left( 2\ln \ln m\right) ^{1/2}$ and $D_{m}=2\ln \ln m+\frac{1}{2%
}\ln \ln \ln m-\frac{1}{2}\ln \pi $. Note that in (\ref{soglia1}) $c_{\alpha
,m}$ does not depend on $m$, whilst it does in (\ref{soglia2}). Note also that \citet{CSW96}, albeit in a different context, choose $\eta =0$. It
is well known that tests based on $\eta =0$ have the smallest power, which
on the contrary increases as $\eta $ increases (see the discussion in %
\citet{lajos04}). 

In order to derive our main theorem, we also need the following assumptions.

\begin{assumption}
\label{ass-6}It holds that (i) $T_{m}=O\left( m^{\varkappa }\right) $ for
some $\varkappa \geq 1$; (ii) $\lim \inf_{m\rightarrow \infty }\frac{T_{m}}{m%
}>0$; (iii) $T_{m}>\tau +C_{0}m^{1/2+\epsilon }$ for $\epsilon >0$ such that 
$\frac{N^{1-\delta }}{m^{1/2-\epsilon }}\rightarrow C_{1}$.
\end{assumption}

\begin{assumption}
\label{restriction-3}It holds that (i) $\int_{-\infty }^{+\infty }\left\vert
u\right\vert ^{4+2\delta }dF_{\psi }\left( u\right) <\infty $; 
\begin{equation*}
\text{\it{(ii)}}\ \ m^{1/2+\epsilon }\left\{W^{-1}+W\left[h\left( \frac{R}{\widetilde{l}\left(
N,m,R\right) }\right)\right]^{-2} +\left[h\left( \frac{R}{\widetilde{l}%
\left( N,m,R\right) }\right)\right]^{-1} \right\} \rightarrow 0,
\end{equation*}%
for some $\epsilon >0$.
\end{assumption}

Assumption \ref{ass-6} is the same as equation (1.12) in \citet{lajos07},
and it essentially requires that the monitoring goes on for a sufficiently
long time, longer than the initial training period $m$. In particular, we need to monitor for a number of periods of order at least $m^{1/2}$. Assumption \ref%
{restriction-3} strengthens Assumption \ref{ass-5}\textit{(ii)}, and it is
needed to prove a moment condition for the sequence $\{\Gamma
_{t}\}_{t=m}^{T}$ which will enable a Central Limit Theory to hold. Our main result follows.


\begin{theorem}
\label{darling-erdos}Let Assumptions \ref{ass-1}-\ref{restriction-3} hold.
Under $H_{0}$\ it holds that, as $\min \left( N,m,R,W\right) \rightarrow
\infty $ 
\begin{equation}
P^{\dag }\left( \max_{1\leq k\leq T_{m}}\frac{d\left( k;m\right) }{\nu
^{\ast }\left( k;m\right) }\leq x\right) \rightarrow P\left( \sup_{0\leq
t\leq 1}\frac{\left\vert B\left( t\right) \right\vert }{t^{\eta }}\leq
x\right) ,\text{ \ \ for }\eta \in\l .\l [0,\frac 1 2\r.\r),  \label{erdos-1}
\end{equation}%
%
%
%
%
%
%
%
%
%
%
%
%
%
\begin{equation}
P^{\dag }\left( \max_{1\leq k\leq T_{m}}\frac{d\left( k;m\right) }{\nu
^{\ast }\left( k;m\right) }\leq \frac{x+D_{m}}{A_{m}}\right) \rightarrow
e^{-e^{-x}},\text{ \ \ for }\eta =\frac{1}{2},  \label{erdos-2}
\end{equation}%
for almost all realisations of $\left\{ X_{i,t}\text{, }1\leq i\leq N\text{, 
}1\leq t\leq T\right\} $ and $\{ \xi _{j}(t)\text{, }1\leq j\leq R\text{, }%
m\leq t\leq T\} $ and for $x\in\mathbb{R}^+$.

\noindent Under $H_{A,1}$ and $H_{A,2}$, as $\min \left( N,m,R,W\right)
\rightarrow \infty $, and for a given significance level $\alpha$, it holds
that%
\begin{equation}
c_{\alpha ,m}^{-1}\max_{1\leq k\leq T_{m}}\frac{d\left( k;m\right) }{\nu
^{\ast }\left( k;m\right) }\overset{P^{\dag }}{\rightarrow }\infty ,\text{ \
\ for all }\eta \in\left[0,\frac{1}{2}\right],  \label{power-suffcond-2}
\end{equation}
for almost all realisations of $\left\{ X_{i,t}\text{, }1\leq i\leq N\text{, 
}1\leq t\leq T\right\} $ and $\{ \xi _{j}(t)\text{, }1\leq j\leq R\text{, }%
m\leq t\leq T\} $, where $c_{\alpha ,m}$ is defined in \eqref{soglia1} when $\eta<\frac 12$ and
in \eqref{soglia2} when $\eta=\frac 12$.
\end{theorem}

The main implication of Theorem \ref{darling-erdos} is summarized in the following result (recall that $T=T_m-m$):

\begin{corollary}
\label{corollaryerdos} Under the assumptions of Theorem \ref{darling-erdos}
it holds that, as $\min \left( N,m,R,W\right)\rightarrow \infty $
\begin{align}
& P^{\dag }\left( \widehat{\tau }%
_{m}<T\right) \leq \alpha ,\text{ \ \ under }H_{0},  \label{size} \\
& P^{\dag }\left(t_{N,m}^{\ast }\le  \widehat{\tau }%
_{m}< t^{\ast\ast}_{N,m}\right) =1,\text{ \ \ under }H_{A,1}  \label{power}\\
& P^{\dag }\left(t_{N,m}^{\ast }\le  \widehat{\tau }%
_{m}\le T\right) =1,\text{ \ \ under }H_{A,2},  \label{power2}
\end{align}%
for almost all realisations of $\left\{ X_{i,t}\text{, }1\leq i\leq N\text{, 
}1\leq t\leq T\right\} $ and $\{ \xi _{j}(t)\text{, }1\leq j\leq R\text{, }%
m\leq t\leq T\} $.
\end{corollary}

The notion of size implied by (\ref{size}), in this context, is very
different from the one usually considered in the literature. The purpose of the procedure is to keep the false
rejection probability as little as possible, and therefore (at a minimum)
below the threshold $\alpha $, rather than making it close to $\alpha $.
This makes the monitoring procedure different from the standard
Neyman-Pearson paradigm (and, in general, from a multiple testing exercise): given that the monitoring horizon keeps expanding,
the purpose of $c_{\alpha ,m}$ is to ensure that the chance of a false break
detection is as little as possible -- see also similar comments in %
\citet{lajos07}. 

\subsection{Delay in change-point detection\label{delay}}

A consequence of our approach is that monitoring for a structural change
(despite being in a high-dimensional set-up) can be treated as in a
classical time series framework. In particular, in addition to the
consistency of the procedure, a natural question is how much would the delay
be in detecting a break. In order to formally address this issue, one can
directly use the results by \citet{aue2004delay}; hereafter, we provide a
heuristic discussion of the magnitude of the delay within our setup.

Consider the notation $a_{n}=\Omega \left( b_{n}\right) $ to indicate that the magnitude of the sequence $a_{n}$ is not smaller than that of $b_{n}$%
, viz. $a_{n}>Cb_{n}>0$. Then, by construction, $\{\Gamma _{t}\}_{t=m}^{T}$
has, under the alternative, a \textquotedblleft large\textquotedblright\
shift in the mean after $t_{N,m}^{\ast }$, where $t_{N,m}^{\ast }$ is such
that (recall \eqref{delaytstar})
\begin{equation}
t_{N,m}^{\ast }-\tau =\Omega \left( \frac{m}{N^{1-\delta }}\right) .
\label{delay-1}
\end{equation}%
Defining $\beta $ such that $N=m^{\beta }$, and using (\ref{delta}), it is
possible to analyse (\ref{delay-1}) for various relative rates of divergence
of $m$ and $N$ as they pass to infinity. When $\beta >\frac{1}{2}$, we have
that $\delta =1-\frac{1}{2\beta }+\epsilon $ for an arbitrarily small value
of $\epsilon$. Thus, by (\ref{delay-1})%
\begin{equation*}
t_{N,m}^{\ast }=\tau +\Omega \left( \frac{m}{m^{\beta \left( 1-\delta
\right) }}\right) =\tau +\Omega \left( m^{1/2+\epsilon ^{\prime }}\right) ,
\end{equation*}%
where $\epsilon ^{\prime }>0$ is arbitrarily small. Thus, when $N$ is not
much smaller than $m$, or even larger, the change-point is detected with a delay, $%
t_{N,m}^{\ast }-\tau $, which is of order at least $m^{1/2}$. By the same token,
whenever $\beta \leq \frac{1}{2}$, i.e. $N$ is much smaller than $m$, we have that $\delta =\epsilon $ for an
arbitrarily small value of $\epsilon $, so that 
\begin{equation*}
t_{N,m}^{\ast }=\tau +\Omega \left( m^{1-\beta \left( 1-\delta \right)
}\right) ,
\end{equation*}%
and, by elementary arguments, it follows that $m^{1-\beta ( 1-\delta
) }=\Omega ( m^{1/2+\epsilon ^{\prime }}) $: the delay, in
this case, might be bigger. This is in line with the intuition that a break
will cause $\Gamma _{t}$ - and consequently the detector - to diverge as
fast as $N$: the lower $N$, the lower the divergence rate, and the less
effective the detecion of breaks. Finally, it is interesting to consider the
ultra high-dimensional case, $N=\exp \left( m\right) $. By (\ref{delta}), it
holds that $\delta =1-\left( 1-\epsilon \right) \frac{\ln m}{2m}$ for an
arbitrarily small value of $\epsilon $. Hence, (\ref{delay-1}) yields%
\begin{equation*}
t_{N,m}^{\ast }=\tau +\Omega \left( \frac{m}{\exp \left( \left( 1-\delta
\right) m\right) }\right) =\tau +\Omega \left( m^{1/2+\epsilon ^{\prime
}}\right) ,
\end{equation*}%
again. In essence, in all cases considered there is a delay in the detection
of breaks which is greater than $C_{0}m^{1/2}$, but smaller than $C_{0}m$ -
that is, rescaling the delay by the sample size, this vanishes.

\section{Applying the test under general circumstances\label{discussion}}

The purpose of this section is to discuss how
the test could be applied under slightly different assumptions than the ones
above, and up to which extent such assumptions can be relaxed. More
substantive extensions, which involve modifications of the test, are briefly
discussed in the concluding remarks in Section \ref{conclusions}.

\subsection{Weak factors and local alternatives\label{weak}}

The theory developed in this paper - starting from Assumption \ref{ass-2} -
implicitly requires that, when a new factor appears as a consequence of a
break, this should be a pervasive factor. Indeed, part \textit{(i)} of the
assumption entails that spiked eigenvalues must diverge at a rate $N$, i.e.
a \textquotedblleft strong\textquotedblright\ factor model. However, the
literature has also considered cases in which one or more common factor may
be less pervasive, thus leading to a covariance matrix which has some
eigenvalues passing to infinity at a rate $N^{\kappa }$, for $\kappa \in
\left( 0,1\right) $. A possible example of weak factors arises when
considering jointly macroeconomic data of different countries: global
factors are strong since they are likely to affect all countries; however
national factors, although strong within a given country, will affect only a
subset of all variables considered and can be seen as weak -- see e.g. the
empirical study in \citet{moench_dynamic_2013}. Estimation of factor models
in the presence of such \textquotedblleft weak\textquotedblright\ or
``local'' factors have been paid considerable attention by the literature -
see \citet{DGR08}, \citet{onatski12}, in the same setting as ours and, in a
slightly different context, \citet{LY12}.
%
The notion of weak factors is intertwined with that of a local alternative
hypothesis where the break does happen but it is \textquotedblleft
small\textquotedblright, for example when a break is caused by a change of only some, but not all, loadings. 
We focus on the (algebraically simpler) case of $H_{A,2}$. Consistently with the literature on weak factors,
we allow the $\left( r+1\right) $%
-th eigenvalue to behave as 
\begin{equation}  \label{weakbreak}
\lambda ^{\left( r+1\right) }\left( t\right) = C_{0}N^{\kappa}, \text{\ \ for }\ \tau \leq t\leq T,
\end{equation}
for $\kappa\in(0,1)$, while it is bounded for all other values of $t$.

We now discuss heuristically under which conditions such small breaks can be
detected; we consider for simplicity the case $\eta <\frac{1}{2}$. We know
that, based on Theorem \ref{gamma}, a break in the $\left( r+1\right) $-%
th largest eigenvalue enters the sequence $\{\Gamma _{t}\}_{t=m}^{T}
$ as a shift in its mean: this is essentially the way in which the
monitoring procedure picks up the presence of a break. In particular, by analysing the proof of
Theorem \ref{theta} and using a Mean Value argument, it follows that%
\begin{equation}
\Theta _{t}\approx R\int_{-\infty }^{+\infty }\left\vert G_{\phi }\left(
u\phi _{N,m}^{-1}\left( t\right) \right) -G_{\phi }\left( 0\right)
\right\vert ^{2}dF_{\phi }\left( u\right) \approx C_{0}R\phi
_{N,m}^{-2}\left( t\right) ,  \label{local-break-2}
\end{equation}%
for any $t\geq t_{N,m}^{\ast }$ for which $H_{A,2}$ holds. Then, from \eqref{gamma-alternative}, for the same values of $t\geq t^{\ast}_{N,m}$ for which \eqref{local-break-2} holds we have 
\begin{equation}
\Gamma _{t}\approx W\int_{-\infty }^{+\infty }\left\vert G_{\psi }\left(
u\psi _{N,m,R}^{-1}\left( t\right) \right) -G_{\psi }\left( 0\right)
\right\vert ^{2}dF_{\psi }\left( u\right) \approx C_{0}W\psi
_{N,m,R}^{-2}\left( t\right)  \label{local-break-1}.
\end{equation}%
Consider now the case where $g (\cdot 
)$ in (\ref{phi}) and  $h\left( \cdot \right) $ in (\ref%
{psi-1}) are both the identity function. Recalling the notation $N=m^{\beta }$, and noting that by \eqref{weakbreak} we have $\phi
_{N,m}\left( t\right) \approx N^{\kappa -\delta }$, by \eqref{psi-1}, \eqref{local-break-2}, and \eqref{local-break-1}, we have 
\begin{align}
\Gamma _{t}&
\approx C_0  W R^{-2} (\ln N)^{8+\epsilon}(\ln R)^{4+\epsilon} N^{4(\kappa-\delta)}=\Delta _{N,R,W}. \nonumber
\end{align}
Upon inspecting the proof of Theorem \ref{darling-erdos}, in order for the procedure to detect a break, it is required that $m^{1/2} \Delta_{N,R,W} \rightarrow \infty$, as $\min(m,N,R,W)\to\infty$. 
Therefore, if $\Delta_{N,R,W}\to\infty$ the break is always detectable. If instead $\Delta_{N,R,W}\to 0$, we are in presence of a shrinking break. By Assumption \ref{restriction-2}, a sufficient condition to have a shrinking break is 
\beq
\kappa\le \delta, \label{suffweak}
\eeq
 and a necessary condition for the break to be detectable is 
 \beq
 \kappa>\delta-\frac 1{8\beta}. \label{necweak}
 \eeq

Consider first the case $N>m^{1/2}$. Then, by definition of $\delta$ we always have a shrinking break whenever $\kappa\le1-\frac 1{2\beta}$ and moreover a new weak factor is detected if at least $\kappa > 1-\frac 5{8\beta}$. This entails that we can hope to detect new weak factor for any $\kappa> 0$ only if $\beta < \frac 58$; conversely, for larger values of $\beta$ the range of values of $\kappa$ for which we can detect a new factor is reduced, e.g. for $N=m$, we must have at least $\kappa>\frac 3 8$.

Turning to the case $N\le m^{1/2}$, since we can choose $\delta$ to be infinitesimally small, \eqref{suffweak} is never satisfied but \eqref{necweak} is always satisfied and in general we cannot say more about the ability of our procedure to detect a shrinking break. However, we note that in the case the case $N=R=W$, as in Sections \ref{numerics} and \ref{emp} below, a necessary and sufficient condition for a break to be shrinking and detectable is $\delta+\frac 14-\frac 1{8\beta}  <\kappa<\delta+\frac 14$, and when $N\le m^{1/2}$ a new factor is always detected regardless of $\kappa$.

\subsection{Heteroskedasticity in the idiosyncratic component\label%
{hetoeroskedasticity}}

The main assumptions in the paper are spelt out with respect to $X_{i,t}$,
avoiding to make any comments on the properties of $u_{i,t}$ across time. We
now discuss the behaviour of the test in the presence of heteroskedasticity,
which is not explicitly considered (although not ruled out) by Assumption %
\ref{ass-2}. For the sake of simplicity, we consider the case of an abrupt
change in the covariance matrix of $\{u_{i,t}\}_{i=1}^N$, although more
general forms of heteroskedasticity could also be considered.

To illustrate this, we consider a simple example where the covariance matrix 
$E\left( u_{t}u_{t}^{\prime }\right)$ undergoes an abrupt change of size $%
\Delta _{u}$ after a point in time, say $\tau ^{\ast }$:%
\begin{equation*}
E\left( u_{t}u_{t}^{\prime }\right) =\left\{ 
\begin{array}{l}
\Sigma _{u} \\ 
\Sigma _{u}+\Delta _{u}%
\end{array}%
\right. \text{ for }%
\begin{array}{l}
m\leq t< \tau ^{\ast }, \\ 
\tau ^{\ast }\leq t\leq T,%
\end{array}%
\end{equation*}%
where $\Delta_u$ affects some or even all covariances. The only condition we require in order for our test to be applicable is $%
\omega ^{\left( 1\right) }\left( m^{-1}\sum_{k=t-m+1}^{t}E\left(
u_{t}u_{t}^{\prime }\right) \right) \leq C_0$ for each $t\ge m$, where the
notation $\omega ^{\left( 1\right) }\left( A\right) $\ is understood to
represent the largest eigenvalue of a matrix $A$. This holds, when $t< \tau
^{\ast }$, as long as $\omega ^{\left( 1\right) }\left( \Sigma _{u}\right)
\leq C_0$. When $t\geq \tau ^{\ast }$, using Weyl's inequality it follows
that%
\begin{equation}
\omega ^{\left( 1\right) }\left( \frac 1 m\sum_{k=t-m+1}^{t}E\left(
u_{t}u_{t}^{\prime }\right) \right) \leq \omega ^{\left( 1\right) }\left(
\Sigma _{u}\right) +\omega ^{\left( 1\right) }\left( \Delta _{u}\right) ,
\label{negligible-break-idio}
\end{equation}%
which is bounded as long as $\omega ^{\left( 1\right) }\left( \Sigma
_{u}\right) \leq C_{0}$ and $\omega ^{\left( 1\right) }\left( \Delta
_{u}\right) \leq C_{1}$. In essence, as long as the perturbation matrix $%
\Delta _{u}$ is not too big, and therefore as long as the changes in the
covariance structure of the idiosyncratic are not too big, our test can
still be applied.

Condition (\ref{negligible-break-idio}) has interesting implications.
Consider a break such that $\Delta _{u}=diag\left\{ d_{i}\right\} $, with $%
0\leq d_{i}\leq C_{0}$ for all $1\le i\le N$. In such a case, where the
variances of the error terms all undergo a change (potentially), but the
covariance structure does not change, it would hold that $\omega ^{\left(
1\right) }\left( \Delta _{u}\right) \leq C_{0}$: even a large (but of finite
size) break in the variance of the idiosyncratic components does not alter
the structure of the eigenvalues of $E\left( X_{t}X_{t}^{\prime }\right) $,
by introducing a spurious spiked eigenvalue. Thus, an interesting question
about the robustness of our procedure is: when is a break in the
idiosyncratic component strong enough to be confused with a break in the
factor structure? By the same (heuristic) token as above, the eigenvalue
structure of $E\left( X_{t}X_{t}^{\prime }\right) $ would change if, for
argument's sake, $\omega ^{\left( 1\right) }\left(
m^{-1}\sum_{k=t-m+1}^{t}E\left( u_{t}u_{t}^{\prime }\right) \right)
=C_{0}N^{\varepsilon }$ with $\varepsilon \in \left( 0,1\right] $. By Weyl's
inequality assuming for simplicity that there is no break in the factor
component%
\begin{equation}
\omega ^{\left( 1\right) }\left( \frac1 m\sum_{k=t-m+1}^{t}E\left(
u_{t}u_{t}^{\prime }\right) \right) \geq \omega ^{\left( N\right) }\left(
\Sigma _{u}\right) +\omega ^{\left( 1\right) }\left( \Delta _{u}\right) \geq
\omega ^{\left( 1\right) }\left( \Delta _{u}\right) .
\label{non-negligible-break-idio}
\end{equation}%
Therefore, a sufficient condition would be $\omega ^{\left( 1\right) }\left(
\Delta _{u}\right) =C_{0}N^{\varepsilon }$. Moreover, given that $\omega
^{\left( 1\right) }\left( \Delta _{u}\right) \geq N^{-1}
\sum_{i=1}^N\sum_{j=1}^N\left\{ \Delta _{u}\right\} _{i,j}, $ then (\ref%
{non-negligible-break-idio}) suggests that a break which is
\textquotedblleft sufficiently pervasive\textquotedblright , so that it
affects not merely the variances of the idiosyncratic components, but also
their covariances (without needing to be necessarily huge), could introduce
a spiked eigenvalue in $E\left( X_{t}X_{t}^{\prime }\right)$. In such cases
our procedure might detect $\tau^*$ as a change-point even if the signal
component does not change - see also the same phenomenon documented in the off-line case by \citet{BCF16}.

\subsection{Extensions to consider further alternative hypotheses\label%
{alternatives}}

So far, we have focused our attention onto two empirically relevant but very
specific forms of alternative hypotheses: a possible change in the
loadings - $H_{A,1}$ - and a possible increase in the number of factors - $%
H_{A,2}$. However, our methodology is sufficiently general to be adapted
(with minor modifications) to other cases also. A\ leading example is the
case in which $q\ge 1$ factors vanish, viz.%
\begin{equation}
H_{A,3}: \left\{ 
\begin{array}{l}
X_{i,t}=\sum_{j=1}^{r}a_{ij}f_{jt}+u_{i,t} \\ 
X_{i,t}=\sum_{j=1}^{r-q}\widetilde{a}_{ij}f_{jt}+u_{i,t}%
\end{array}%
\text{ for }%
\begin{array}{l}
1\leq t< \tau \\ 
\tau \leq t\leq T%
\end{array}%
\right. .  \label{ha-extra-1}
\end{equation}%
Note that, in (\ref{ha-extra-1}), we can entertain the possibility that the
loadings of the non-vanishing factors may also be subject to changes,
although this is not required. For simplicity consider the case $q=1$, then
under (\ref{ha-extra-1}), it can be noted that the $r$-th eigenvalue of the
covariance matrix of $X_{i,t}$ is spiked before $\tau $, and bounded
thereafter. This suggests that testing for (\ref{ha-extra-1}) can be based
on $\widehat{\lambda }^{\left( r\right) }\left( t\right) $. Since under the
null (in essence, on account of Lemma \ref{rolling-eigenvalues-2}) $%
N^{-\delta }\widehat{\lambda }^{\left( r\right) }\left( t\right) \rightarrow
\infty $, whereas under the alternative $N^{-\delta }\widehat{\lambda }%
^{\left( r\right) }\left( t\right) \rightarrow 0$, one round of
randomisation is enough to have a sequence of test statistics which behaves
like $\{\Gamma _{t}\}_{t=m}^T$ under the null - that is, which (conditional
on the sample) is \textit{i.i.d.}, has moments that exist up to any order,
and has an asymptotic chi-square distribution, with mean and variance that
can be approximated with a polynomially vanishing error. Hence, monitoring
can be again carried out as proposed in Section \ref{monitoring}.

\section{Monte Carlo simulations}\label{numerics}


Under $H_{0}$ we simulate data according to the stable factor model %
\eqref{fmscalar2}:
\begin{equation*}
X_{i,t}=a_{i}^{\prime }f_{t}+u_{i,t}, \quad 1\le i\le N,\; 1\le t\le T.
\end{equation*}%
In particular, we fix $N=100$, and we consider $r\in \{1,2,3,4\}$ factors.
As far as the time dimension is concerned, we consider burn-in periods and thus sample sizes of
dimension $m\in \{50,75,100,125,150,175,200,225,250\}$. 
We monitor our model for $1000$ periods (that is, we set $T=1000$). We
simulate each element of the loadings vector $a_{i}$ as $i.i.d.\mathcal N\left(
0,1\right) $; we assume some time dependence in the common factors through
a causal VAR(1) process%
\begin{equation*}
f_{t}=Hf_{t-1}+e_{t},\quad 1\le t\le T,
\end{equation*}%
where $e_{t}\sim i.i.d.\mathcal N\left( 0,I_{r}\right) $ and the matrix $H$ has
maximum absolute value of the eigenvalues equal to $0.7$. The $N\times T$
matrix of idiosyncratic components $u$ is generated as $u =D\varepsilon G$,
where the $NT\times 1$ vector of stacked columns of $\varepsilon$ is $%
i.i.d.\mathcal N\left( 0,I_{NT}\right) $ and $D$ and $G$ are two $N\times N$ and $%
T\times T$ Toeplitz matrices with entries, in the $k$-th diagonal
place, given by $0.3^{k-1}$ and $0.5^{k-1}$ respectively. Finally, we have
set the signal-to-noise ratio to $\frac{Var\left( X_{i,t}\right)}{Var\left(
u_{i,t}\right)}=2$ for all $1\le i\le N$.

Under the alternative, we consider breaks to occur at the change-point $\tau
=500$ under the two schemes:
\begin{eqnarray}
&&X_{i,t} =a_{i}^{\prime }f_{t}\ I[t< \tau]+\widetilde a_{i}^{\prime }f_{t}\
I[t\geq \tau ]+u_{i,t}, \quad 1\le i\le N,\; 1\le t\le T,  \label{scheme-1} \\
&&X_{i,t} =a_{i}^{\prime }f_{t}+b_ig_{t}\ I[t\geq \tau]+u_{i,t}, \quad\quad
\quad\quad\quad \,1\le i\le N,\; 1\le t\le T.  \label{scheme-2}
\end{eqnarray}%
In (\ref{scheme-1}), we consider the case in which all loadings undergo a
change, i.e. $H_{A,1}$; all the elements of $a_{i}$ and $\widetilde a_{i}$
are generated as $i.i.d.\mathcal N\left( 0,1\right) $. Scheme (\ref{scheme-2}) refers
to a break owing to a new common factor, $g_{t}$, appearing, i.e. $H_{A,2}$;
the loadings $b_i$ are generated as $i.i.d. \mathcal N\left( 0,1\right) $, and we
simulate $g_{t}$ as the causal AR(1) 
\begin{equation*}
g_{t}=\varphi g_{t-1}+v_{t},\quad 1\le t\le T,
\end{equation*}%
with $v_{t}\sim i.i.d.\mathcal N\left( 0,1\right) $ and $\varphi=0.7$. The idiosyncratic components are generated as  before.

All results of the test are computed when setting $\eta =0.45$ and $\eta
=0.5 $. The critical values used in the case $\eta =0.45$ are taken from %
\citet{lajos04}; in particular, when the significance level is $\alpha =0.05$
the critical value is $c_{0.05}=2.7992$ and when $\alpha =0.1$ we have $%
c_{0.1}=2.5437$. Regarding the double randomisation, we choose the functions $g(\cdot)$ in \eqref{phi} and $h(\cdot)$ in (\ref{psi-1}) to be the identity, we set $W=R=N$, the distributions $G_{\phi}$ and $G_{\psi}$ in steps $A1$ and $B1$ are chosen to be standard normals, while $F_{\phi}$ and $F_{\psi}$ in steps $A2$ and $B2$ are chosen to have non-zero and equal mass at $\pm \sqrt 2$. 

In order to evaluate the performance of our procedure, we repeat simulations 
$500$ times, and we consider a series of indicators.

\begin{table}[t!]
\begin{center}
\caption{Empirical size - 5\% and 10\% significance}\label{tab:size}
\vskip -.4cm
\phantom{Table1: Power - loadings} Fraction of detections in $[m+1,T]$\vskip .1cm
\small{
\begin{tabular}{c cc|cc || cc|cc ||  cc|cc}
\hline
\hline
& \multicolumn{4}{c||}{$m=50$}& \multicolumn{4}{c||}{$m=75$}& \multicolumn{4}{c}{$m=100$}\\
\hline
& \multicolumn{2}{c|}{$\eta = 0.45$} & \multicolumn{2}{c||}{$\eta=0.5$}& \multicolumn{2}{c|}{$\eta = 0.45$} & \multicolumn{2}{c||}{$\eta=0.5$}  & \multicolumn{2}{c|}{$\eta = 0.45$} & \multicolumn{2}{c}{$\eta=0.5$} \\ 
$r$ & 5\% & 10\% & 5\% & 10\%& 5\% & 10\% & 5\% & 10\%& 5\% & 10\% & 5\% & 10\% \\ \hline
1 & 0.03 & 0.05 & 0.03 & 0.06 & 0.04 & 0.05 & 0.03 & 0.05 & 0.04 & 0.06 & 0.03 & 0.06\\ 
2 & 0.04 & 0.05 & 0.04 & 0.06 & 0.03 & 0.04 & 0.02 & 0.05 & 0.04 & 0.06 & 0.04 & 0.06\\ 
3 & 0.03 & 0.05 & 0.03 & 0.05 & 0.03 & 0.05 & 0.03 & 0.06 & 0.04 & 0.06 & 0.04 & 0.06\\ 
4 & 0.03 & 0.05 & 0.03 & 0.06 & 0.02 & 0.05 & 0.02 & 0.06 & 0.04 & 0.05 & 0.03 & 0.06\\ 
\hline
\hline
& \multicolumn{4}{c||}{$m=125$}& \multicolumn{4}{c||}{$m=150$}& \multicolumn{4}{c}{$m=175$}\\
\hline
& \multicolumn{2}{c|}{$\eta = 0.45$} & \multicolumn{2}{c||}{$\eta=0.5$}& \multicolumn{2}{c|}{$\eta = 0.45$} & \multicolumn{2}{c||}{$\eta=0.5$}  & \multicolumn{2}{c|}{$\eta = 0.45$} & \multicolumn{2}{c}{$\eta=0.5$} \\ 
$r$ & 5\% & 10\% & 5\% & 10\%& 5\% & 10\% & 5\% & 10\%& 5\% & 10\% & 5\% & 10\% \\ \hline
1 & 0.05 & 0.06 & 0.05 & 0.06 & 0.04 & 0.05 & 0.03 & 0.05 & 0.04 & 0.07 & 0.04 & 0.06\\ 
2 & 0.03 & 0.05 & 0.03 & 0.06 & 0.04 & 0.05 & 0.03 & 0.05 & 0.05 & 0.06 & 0.04 & 0.07\\ 
3 & 0.03 & 0.05 & 0.03 & 0.05 & 0.03 & 0.05 & 0.03 & 0.05 & 0.04 & 0.08 & 0.04 & 0.08\\ 
4 & 0.03 & 0.06 & 0.03 & 0.07 & 0.04 & 0.07 & 0.05 & 0.08 & 0.04 & 0.06 & 0.05 & 0.06\\ 
\hline
\hline
& \multicolumn{4}{c||}{$m=200$}& \multicolumn{4}{c||}{$m=225$}& \multicolumn{4}{c}{$m=250$}\\
\hline
& \multicolumn{2}{c|}{$\eta = 0.45$} & \multicolumn{2}{c||}{$\eta=0.5$}& \multicolumn{2}{c|}{$\eta = 0.45$} & \multicolumn{2}{c||}{$\eta=0.5$}  & \multicolumn{2}{c|}{$\eta = 0.45$} & \multicolumn{2}{c}{$\eta=0.5$} \\ 
$r$ & 5\% & 10\% & 5\% & 10\%& 5\% & 10\% & 5\% & 10\%& 5\% & 10\% & 5\% & 10\% \\ \hline
1 & 0.05 & 0.07 & 0.05 & 0.07 & 0.05 & 0.07 & 0.05 & 0.07 & 0.04 & 0.07 & 0.04 & 0.08\\ 
2 & 0.04 & 0.05 & 0.03 & 0.05 & 0.04 & 0.06 & 0.04 & 0.07 & 0.04 & 0.07 & 0.04 & 0.07\\ 
3 & 0.05 & 0.07 & 0.04 & 0.08 & 0.03 & 0.04 & 0.03 & 0.04 & 0.03 & 0.05 & 0.04 & 0.05\\ 
4 & 0.04 & 0.06 & 0.04 & 0.06 & 0.04 & 0.06 & 0.04 & 0.07 & 0.04 & 0.07 & 0.04 & 0.08\\ \hline
\end{tabular}
}
\end{center}
\end{table}

\begin{enumerate}
\item[(1)] In Table \ref{tab:size} we report the fraction of false
rejections over the whole monitoring period ($m+1\le t\le T$), when no break
is present, i.e. under $H_0$, and when testing at 5\% and 10\% significance
levels. As expected the empirical size is always below the significance
level.

\begin{table}[t!]
\begin{center}
\caption{Power - loadings change - 5\% significance}\label{tab:power1}
\vskip -.4cm
\phantom{Table2: Power - loadings} Fraction of detections in $[\tau,\tau+m)$\vskip .1cm
\small{
\begin{tabular}{c | cccc ccccc}
\hline
\hline
$\eta=0.45$&\multicolumn{9}{c}{$m$}\\
$r$&50&75&100&125&150&175&200&225&250\\
\hline
1&0.96	&	0.95	&	0.96	&	0.96	&	0.96	&	0.96	&	0.96	&	0.95	&	0.95	\\
2&0.58	&	0.97	&	0.97	&	0.96	&	0.96	&	0.96	&	0.97	&	0.95	&	0.98	\\
3&0.01	&	0.74	&	0.97	&	0.97	&	0.96	&	0.97	&	0.96	&	0.97	&	0.96	\\
4&0.00	&	0.03	&	0.80	&	0.94	&	0.96	&	0.94	&	0.96	&	0.96	&	0.96	\\
\hline
\hline
$\eta=0.5$&\multicolumn{9}{c}{$m$}\\
$r$&50&75&100&125&150&175&200&225&250\\
\hline
1&0.96	&	0.95	&	0.96	&	0.96	&	0.97	&	0.95	&	0.96	&	0.96	&	0.96	\\
2&0.44	&	0.97	&	0.97	&	0.96	&	0.96	&	0.96	&	0.97	&	0.95	&	0.98	\\
3&0.00	&	0.62	&	0.97	&	0.97	&	0.96	&	0.96	&	0.97	&	0.96	&	0.97	\\
4&0.00	&	0.01	&	0.65	&	0.95	&	0.97	&	0.95	&	0.96	&	0.96	&	0.97	\\
\hline
\hline
\end{tabular}
}
\end{center}
\end{table}

\begin{table}[th!]
\begin{center}
\caption{Power - loadings change - 10\% significance}\label{tab:power2}
\vskip -.4cm
\phantom{Table2: Power - loadings} Fraction of detections in $[\tau,\tau+m)$\vskip .1cm
\small{
\begin{tabular}{c | cccc ccccc}
\hline
\hline
$\eta=0.45$&\multicolumn{9}{c}{$m$}\\
$r$&50&75&100&125&150&175&200&225&250\\
\hline
1	&	0.94	&	0.93	&	0.95	&	0.93	&	0.94	&	0.93	&	0.94	&	0.94	&	0.93	\\
2	&	0.66	&	0.95	&	0.95	&	0.93	&	0.95	&	0.93	&	0.94	&	0.94	&	0.95	\\
3	&	0.01	&	0.81	&	0.95	&	0.95	&	0.94	&	0.95	&	0.95	&	0.94	&	0.94	\\
4	&	0.00	&	0.06	&	0.87	&	0.94	&	0.95	&	0.92	&	0.93	&	0.94	&	0.94	\\
\hline
\hline
$\eta=0.5$&\multicolumn{9}{c}{$m$}\\
$r$&50&75&100&125&150&175&200&225&250\\
\hline
1	&	0.94	&	0.93	&	0.95	&	0.92	&	0.93	&	0.93	&	0.93	&	0.93	&	0.93	\\
2	&	0.59	&	0.94	&	0.95	&	0.92	&	0.94	&	0.93	&	0.94	&	0.93	&	0.95	\\
3	&	0.01	&	0.73	&	0.94	&	0.95	&	0.94	&	0.95	&	0.94	&	0.93	&	0.94	\\
4	&	0.00	&	0.03	&	0.80	&	0.93	&	0.94	&	0.91	&	0.93	&	0.94	&	0.93	\\
\hline
\hline
\end{tabular}
}
\end{center}
\end{table}

\begin{table}[t!]
\begin{center}
\caption{Power - new factor appears - 5\% significance}\label{tab:power3}
\vskip -.4cm
\phantom{Table2: Power - loadings} Fraction of detections in $[\tau,\tau+m)$\vskip .1cm
\small{
\begin{tabular}{c | cccc ccccc}
\hline
\hline
$\eta=0.45$&\multicolumn{9}{c}{$m$}\\
$r$&50&75&100&125&150&175&200&225&250\\
\hline
1	&	0.93	&	0.93	&	0.92	&	0.95	&	0.92	&	0.95	&	0.93	&	0.93	&	0.92	\\
2	&	0.78	&	0.96	&	0.97	&	0.95	&	0.95	&	0.96	&	0.96	&	0.95	&	0.96	\\
3	&	0.10	&	0.89	&	0.97	&	0.98	&	0.96	&	0.97	&	0.95	&	0.95	&	0.97	\\
4	&	0.00	&	0.27	&	0.91	&	0.96	&	0.96	&	0.96	&	0.95	&	0.95	&	0.96	\\
\hline
\hline
$\eta=0.5$&\multicolumn{9}{c}{$m$}\\
$r$&50&75&100&125&150&175&200&225&250\\
\hline
1	&	0.95	&	0.94	&	0.93	&	0.96	&	0.95	&	0.96	&	0.95	&	0.93	&	0.94	\\
2	&	0.71	&	0.96	&	0.97	&	0.96	&	0.95	&	0.96	&	0.96	&	0.96	&	0.97	\\
3	&	0.06	&	0.85	&	0.97	&	0.98	&	0.96	&	0.98	&	0.96	&	0.96	&	0.98	\\
4	&	0.00	&	0.16	&	0.89	&	0.96	&	0.96	&	0.95	&	0.95	&	0.96	&	0.97	\\
\hline
\hline
\end{tabular}
}
\end{center}
\end{table}

\begin{table}[th!]
\begin{center}
\caption{Power - new factor appears - 10\% significance}\label{tab:power4}
\vskip -.4cm
\phantom{Table2: Power - loadings} Fraction of detections in $[\tau,\tau+m)$\vskip .1cm
\small{
\begin{tabular}{c | cccc ccccc}
\hline
\hline
$\eta=0.45$&\multicolumn{9}{c}{$m$}\\
$r$&50&75&100&125&150&175&200&225&250\\
\hline
1	&	0.89	&	0.88	&	0.88	&	0.90	&	0.88	&	0.91	&	0.92	&	0.88	&	0.88	\\
2	&	0.81	&	0.94	&	0.95	&	0.93	&	0.91	&	0.93	&	0.93	&	0.93	&	0.93	\\
3	&	0.14	&	0.88	&	0.94	&	0.95	&	0.94	&	0.95	&	0.92	&	0.93	&	0.95	\\
4	&	0.00	&	0.36	&	0.92	&	0.94	&	0.94	&	0.93	&	0.93	&	0.93	&	0.94	\\
\hline
\hline
$\eta=0.5$&\multicolumn{9}{c}{$m$}\\
$r$&50&75&100&125&150&175&200&225&250\\
\hline
1	&	0.90	&	0.89	&	0.88	&	0.91	&	0.89	&	0.92	&	0.91	&	0.89	&	0.89	\\
2	&	0.76	&	0.94	&	0.95	&	0.92	&	0.92	&	0.93	&	0.93	&	0.93	&	0.93	\\
3	&	0.10	&	0.86	&	0.94	&	0.96	&	0.94	&	0.95	&	0.93	&	0.93	&	0.94	\\
4	&	0.00	&	0.27	&	0.89	&	0.94	&	0.93	&	0.93	&	0.92	&	0.92	&	0.94	\\
\hline
\hline
\end{tabular}
}
\end{center}
\end{table}

\item[(2)] In Tables \ref{tab:power1}, \ref{tab:power2}, \ref{tab:power3}
and \ref{tab:power4} we show the fraction of detections for which $\tau\le \widehat {\tau}_m<\tau+m-1$, when a break takes place under $H_{A,1}$ or $H_{A,2}$ and when testing at 5\% and 10\% significance levels, setting either $\eta=0.45$ or $\eta=0.5$. Results show that the test does have power versus the two
alternative hypotheses considered in this paper. As the construction of the
test and the theory would suggest, the power declines as $r$, the original,
pre-break number of factors, increases: in essence, the test checks whether
an eigenvalue is large, and the magnitude of the $(r+1)$-th largest
eigenvalue declines with $r$. Still, even when $r=4$, the test has high
power when $m\geq 100$ in all cases considered, and, in presence of a new
factor appearing (see Tables \ref{tab:power3} and \ref{tab:power4}), even
when $m\geq 50$. An interesting feature of the test is the case $\eta =0.5$:
although in theory this choice yields the highest power, it is well known
that convergence to the extreme value distribution is very slow, leading to
larger than correct critical values, and, consequently, to lower power (see
the comments in \citet{csorgo1997}). However, considering the discrepancy
between the power when $\eta =0.45$ and $\eta =0.5$, this is not always the case: tests based on the choice $\eta =0.5$ have roughly the same power as for the case $\eta =0.45$ whenever there is a
change in the loadings, and also when there is a new factor appearing (at
least for a sample size $m\geq 100$). Last, notice that when considering $H_{A,2}$ (a new factor appearing), then we could also detect a change-point when $\tau+m-1\le t\le T$, but we do not report results in this case since power can only increase with respect to what shown in Tables \ref{tab:power3} and \ref{tab:power4}.

\item[(3)] In Tables \ref{tab:loc1} and \ref{tab:loc2} we report the
minimum, maximum, the 25$^{\text{{\tiny th}}}$, 50$^{\text{{\tiny th}}}$ and
75$^{\text{{\tiny th}}}$ percentiles of the distribution of the estimated
change-point locations, whenever under $H_{A,1}$ or $H_{A,2}$ a break is
detected at $\widehat {\tau}_m$ such that $\tau\le \widehat {\tau}_m \le T$
and when testing at 10\% significance levels. Given the results in Tables %
\ref{tab:power2} and \ref{tab:power4} we report those statistics only for $%
m=100,175,250$. It is evident that the test detects a break with a delay
which increases as $r$ increases - this is in line with the comments in
Section \ref{weak}, since, as $r$ grows, the $r$-th eigenvalue becomes
smaller and smaller, thus being closer to a weak factor. Interestingly,
there are virtually no differences between the cases of $\eta =0.45$ and $%
\eta =0.5$; similarly, different values of $m$ also do not seem to alter
results. Note that, as expected, the minimum values of the distribution of the estimated locations are, roughly speaking, of order $m^{1/2}$ all across the table.

\begin{table}[t!]
\caption{Location distribution - loadings change }
\label{tab:loc1}
{\ \centering $\phantom{\qquad\qquad\qquad\qquad\qquad%
\qquad}$ {(true change-point at $\tau=500$)} \newline
}
\begin{center}
{\small {\ \vskip -1cm 
\begin{tabular}{llccccc|ccccc}
\hline\hline
&  & \multicolumn{5}{c|}{$\eta = 0.45$} & \multicolumn{5}{c}{$\eta = 0.5$}
\\ 
$m$ & $r$ & min & 25$^{\mbox{th}}$ & 50$^{\mbox{th}}$ & 75$^{\mbox{th}}$ & 
max & min & 25$^{\mbox{th}}$ & 50$^{\mbox{th}}$ & 75$^{\mbox{th}}$ & max \\ 
\hline
100 & 1 & 504 & 516 & 520 & 526 & 551 & 505 & 516 & 521 & 526 & 551 \\ 
& 2 & 507 & 523 & 528 & 533 & 553 & 507 & 524 & 529 & 534 & 554 \\ 
& 3 & 509 & 529 & 536 & 543 & 573 & 509 & 531 & 538 & 545 & 586 \\ 
& 4 & 520 & 548 & 559 & 571 & $> T$ & 520 & 553 & 563 & 578 & $> T$ \\ \hline
175 & 1 & 504 & 516 & 522 & 528 & 550 & 505 & 517 & 523 & 529 & 550 \\ 
& 2 & 508 & 523 & 529 & 535 & 553 & 509 & 524 & 530 & 537 & 555 \\ 
& 3 & 507 & 531 & 537 & 543 & 564 & 507 & 532 & 538 & 544 & 566 \\ 
& 4 & 514 & 536 & 544 & 551 & 580 & 514 & 538 & 546 & 553 & 591 \\ \hline
250 & 1 & 502 & 518 & 523 & 529 & 551 & 505 & 519 & 524 & 529 & 553 \\ 
& 2 & 508 & 525 & 531 & 537 & 558 & 508 & 526 & 532 & 538 & 565 \\ 
& 3 & 508 & 531 & 538 & 547 & 568 & 508 & 533 & 540 & 548 & 570 \\ 
& 4 & 516 & 538 & 545 & 554 & 581 & 516 & 539 & 547 & 555 & 584 \\ 
\hline\hline
\end{tabular}%
} }
\end{center}
\end{table}

\begin{table}[th!]
\caption{Location distribution - new factor appears}
\label{tab:loc2}{\ \centering $\phantom{\qquad\qquad\qquad\qquad\qquad%
\qquad}$ {(true change-point at $\tau=500$)}\newline
}
\begin{center}
{\small {\ \vskip -1cm 
\begin{tabular}{llccccc|ccccc}
\hline\hline
&  & \multicolumn{5}{c|}{$\eta = 0.45$} & \multicolumn{5}{c}{$\eta = 0.5$}
\\ 
$m$ & $r$ & min & 25$^{\mbox{th}}$ & 50$^{\mbox{th}}$ & 75$^{\mbox{th}}$ & 
max & min & 25$^{\mbox{th}}$ & 50$^{\mbox{th}}$ & 75$^{\mbox{th}}$ & max \\ 
\hline
100 & 1 & 502 & 514 & 519 & 525 & 547 & 505 & 515 & 520 & 525 & 548 \\ 
& 2 & 505 & 522 & 530 & 537 & 574 & 506 & 523 & 531 & 538 & 576 \\ 
& 3 & 513 & 532 & 541 & 550 & 598 & 514 & 534 & 542 & 552 & 598 \\ 
& 4 & 519 & 546 & 557 & 569 & 648 & 520 & 548 & 560 & 574 & 669 \\ \hline
175 & 1 & 501 & 515 & 520 & 525 & 546 & 502 & 516 & 521 & 526 & 552 \\ 
& 2 & 507 & 524 & 531 & 537 & 564 & 507 & 524 & 532 & 539 & 572 \\ 
& 3 & 508 & 532 & 541 & 550 & 578 & 508 & 534 & 542 & 552 & 579 \\ 
& 4 & 518 & 543 & 552 & 562 & 598 & 518 & 544 & 555 & 565 & 599 \\ \hline
250 & 1 & 502 & 514 & 520 & 527 & 552 & 502 & 514 & 521 & 527 & 553 \\ 
& 2 & 508 & 524 & 531 & 540 & 571 & 508 & 525 & 532 & 541 & 571 \\ 
& 3 & 510 & 533 & 542 & 551 & 590 & 510 & 534 & 543 & 552 & 589 \\ 
& 4 & 515 & 542 & 552 & 562 & 600 & 515 & 544 & 553 & 564 & 603 \\ 
\hline\hline
\end{tabular}%
} }
\end{center}
\end{table}
\end{enumerate}

\section{An application to US industrial production data}\label{emp}

We conclude with an
application to a panel of US industrial production indexes. Specifically, we consider monthly 
growth rates for $%
N=224$ sectorial indices, 
over the period from January 1972 to November 2015, for a total of $%
T=527$ observations. Estimation is based on a sample of size $m=60$, i.e. 5 years.
Analysis of the whole dataset using rolling samples of size $m$ suggests
between one and two factors throughout - this result consistently follows
using different procedures - namely, \citet{trapani17} testing procedure and
the criteria by \citet{baing02} and \citet{ABC10}. Therefore, we run our
sequential testing procedure monitoring the first four factors, thus
accounting both for at most two new factors emerging and for a change in all
loadings. The test is run at 5\% significance level and setting $\eta =0.5$,
hence using the critical values in \eqref{soglia2}.

The monitoring is implemented as follows. We begin at $t=m+1$; once the
first change-point is detected at $\widehat{\tau }_{1}\geq m+1$, we restart
the estimation at $t=\widehat{\tau }_{1}$ and after $m$ periods we restart
monitoring at $t=\widehat{\tau }_{1}+m+1$. In general, given an estimated
change-point $\widehat{\tau }_{j}$, with $j\geq 1$, we restart monitoring by
computing the detector defined in \eqref{eq:detector} which in this case is
defined as 
\begin{equation*}
d(k;m)=\left\vert \sum_{t=\widehat{\tau }_{j}+m+1}^{m+k}\frac{\Gamma _{t}-1}{%
\sqrt{2}}\right\vert ,\quad \widehat{\tau }_{j}+1\leq k\leq T-m.
\end{equation*}%
Therefore, the monitoring window after the $j$-th change-point is of size $T-%
\widehat{\tau }_j-m$ and the estimated change-points $\widehat{\tau }_{j}$
are such that $\widehat{\tau }_{j+1}-\widehat{\tau }_{j}\geq m+1$. We keep
restarting the procedure as long as we have a monitoring window of non-zero
length, that is as long as $T-\widehat{\tau }_{j}>m$; this allows the
possibility for the last change-point to be detected in the interval $T-m\le
t\le T$.


We find evidence of two change-points dated: (i) $\wh{\tau}_1$: August, 1983; and (ii) $\wh{\tau}_2$: March, 2008.
The estimated locations are also shown in Figure \ref{fig:sp100} together
with the joint panel of data. The first estimated change-point ($\widehat{\tau}_1$)
clearly mark the start of the Great Moderation, i.e. a period of decrease in volatility of output and inflation, while the second one ($\widehat{\tau}_2$) takes place at the start of
the US recession marked by the Great Financial Crisis. 
Last we discuss the delay of the estimated change-points. Concerning $\widehat{\tau}_1$, there is a general consensus that the start of the Great Moderation is to be dated in 1983, however a precise date is not available, see for example \citet{stockwatson07}. We note here that if we consider the start of the Great Moderation to coincide with the end of the recession of the early 1980s, then the National Bureau of Economic Research (NBER) dates the start of the expansion of the US business cycle in December, 1982, thus the first change-point is detected with a delay of 4 time-periods. 
Concerning $\widehat{\tau}_2$ the NBER dates the start of the recession in December, 2007, therefore we detect the change-point with a delay of 3 time-points \footnote{See \url{https://www.nber.org/cycles/US_Business_Cycle_Expansions_and_Contractions_20120423.pdf}}.
\begin{figure}[t!]
\caption{Estimated change-point locations for US industrial production indexes}
\label{fig:sp100}
\begin{center}
\vskip -.4cm 
\includegraphics[width=.8\textwidth]{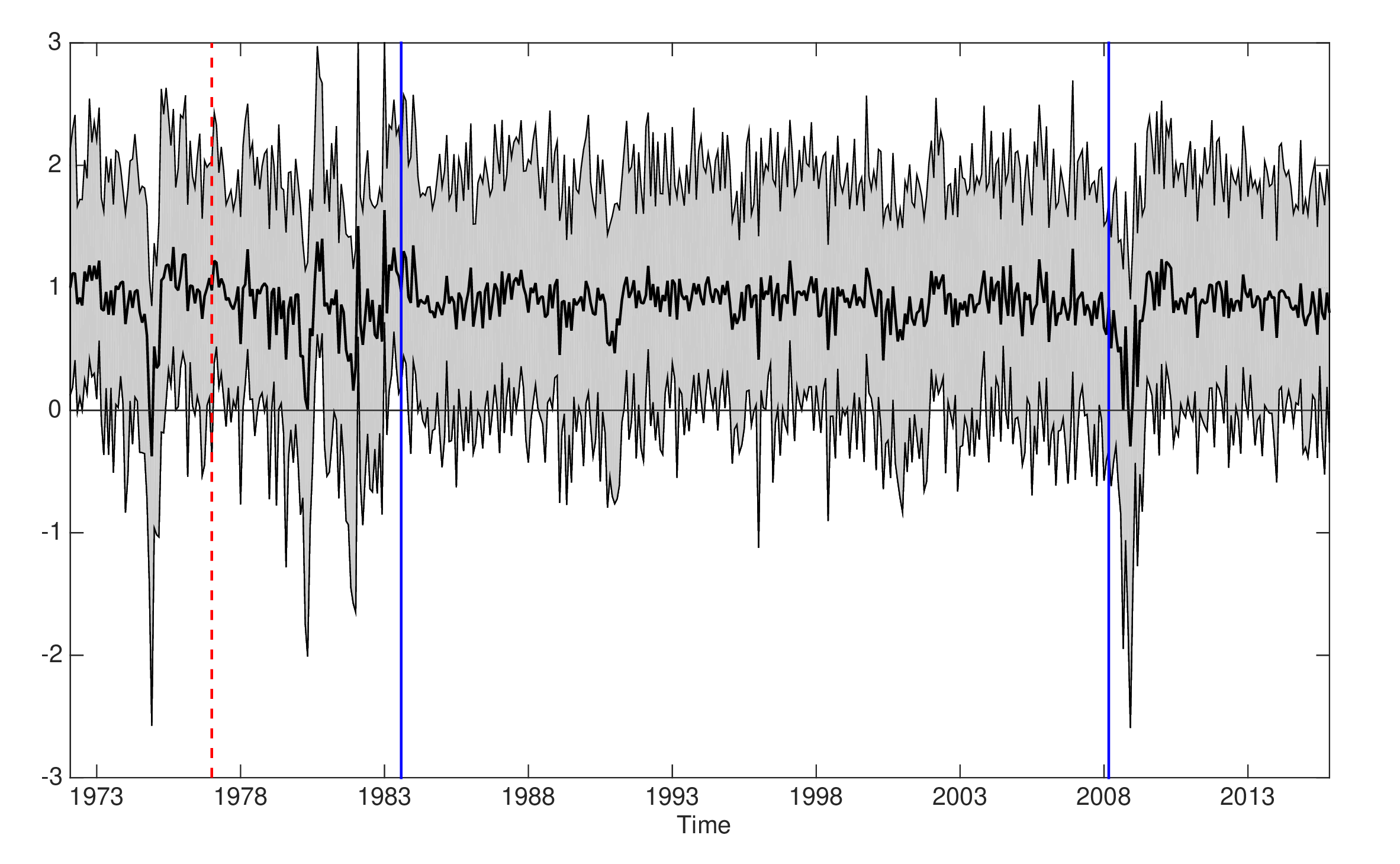} 
\begin{tabular}{p{.9\textwidth}}
{\footnotesize {Black line: cross-sectional median of the data (monthly growth rates on yearly basis); blue lines:
estimated change-point locations; red line: first period used for testing. 
}}%
\end{tabular}%
\end{center}
\par
\vskip -.2cm
\end{figure}

\section{Conclusions\label{conclusions}}

In this paper we develop a a family of monitoring procedures to detect a
break in the signal component of a large factor model; to the best of our
knowledge, this is the first contribution in high-dimensional factor models
which proposes a sequential monitoring and testing procedure, as opposed to
the extant literature where ex-post detection of breaks is usually
considered. Our statistics are based on a well-known property of the $\left(
r+1\right) $-th eigenvalue of the sample covariance matrix of the data:
whilst under the null the $\left( r+1\right) $-th eigenvalue is bounded,
under the alternative of a break (either in the loadings, or in the number
of factors itself) it becomes spiked. Given that the sample eigenvalue does
not have a known limiting distribution under the null, we regularise the
problem by (doubly) randomising the test statistic in conjunction with
sample conditioning, obtaining a sequence of \textit{i.i.d.}, asymptotically
chi-square statistics which are then employed to build the monitoring
scheme. Numerical evidence shows that our procedure works very well in
finite samples, with a very small probability of false detections and tight
detection times in presence of a genuine change-point.

Building on the methodology proposed in this paper, there are at least two
possible extensions which could be considered. Firstly, the results and
methodology in this paper could be also used in the context of a
non-stationary factor model, similar to the one considered in \citet{bai04},
where the factors are allowed to have unit roots. In such case, the key
theoretical result would be to show that in presence of $r$ factors the
first $r$ eigenvalues of the matrix $m^{-2}\sum_{t=1}^{m}X_{t}X_{t}^{\prime
} $ diverge to positive infinity almost surely at some rate, whereas the
remaining factors are a.s. bounded. Secondly, it is possible to extend the
theory developed in this paper to the context of the generalised dynamic
factor model by \citet{FHLR00} or the factor model by \citet{LY12}, which are based on the
asymptotic behavior of the eigenvalues of the spectral density or the long-run covariance matrices, respectively.
By studying the asymptotic behavior of the estimated eigenvalues of those matrices an
appropriate test statistic based on these can be built. These, and other,
extensions are under current investigations by the authors.

%
%


\section*{References}
\bibliography{BT_biblio}
%
%
%
%
%

\newpage

\appendix

\section{Technical appendix}
\setcounter{lemma}{0} \renewcommand{\thelemma}{A.\arabic{lemma}} %
\subsection{Preliminary lemmas}

\label{lemmas}

This section contains technical results which are useful to prove the main
theorems in the paper. Throughout this and the next section, $E^{\ast }$
denotes expectation calculated with respect to $P^{\ast }$; similarly, $%
E^{\dag }$ and $V^{\dag }$\ denote expectation and variance calculated with
respect to $P^{\dag }$. Also, whenever possible, we omit for ease of
notation the dependence of $\xi _{j}\left( t\right) $, and of related
quantities,\ on $t$.

\begin{lemma}
\label{lambda-average}Let%
\begin{equation}
\overline{\lambda }_{N}\left( t\right) = \frac{1}{N}\sum_{p=1}^{N}\lambda
^{\left( p\right) }\left( t\right) \;\mbox{ and }\; \widehat{\overline{\lambda }}_{N}\left( t\right) = \frac{1}{N}\sum_{p=1}^{N}\widehat{\lambda}
^{\left( p\right) }\left( t\right).  \nonumber
\end{equation}%
Under Assumptions \ref{ass-1} and \ref{ass-2}, it holds that%
\begin{equation}
\left\{ 
\begin{array}{l}
\lim \sup_{N\rightarrow \infty }\overline{\lambda }_{N}\left( t\right) =%
\overline{\lambda }^{\sup }\left( t\right) <\infty \\ 
\lim \inf_{N\rightarrow \infty }\overline{\lambda }_{N}\left( t\right) =%
\overline{\lambda }^{\inf }\left( t\right) >0\text{ \ \ }%
\end{array}%
\right. , \nonumber
\end{equation}%
for all $m\leq t\leq T$. Furthermore, under Assumptions \ref{ass-1}-\ref{ass-3},
it holds that\textit{, as }$\min \left\{ N,m\right\} \rightarrow \infty $%
\begin{equation*}
\left\{ 
\begin{array}{l}
\lim \sup_{N\rightarrow \infty }\widehat{\overline{\lambda }}_{N}\left(
t\right) =\overline{\lambda }^{\sup }\left( t\right) <\infty \\ 
\lim \inf_{N\rightarrow \infty }\widehat{\overline{\lambda }}_{N}\left(
t\right) =\overline{\lambda }^{\inf }\left( t\right) >0%
\end{array}%
\right. .
\end{equation*}

\begin{proof}
See Lemmas 2.1 and A.1 in \citet{trapani17}.
\end{proof}
\end{lemma}

\begin{lemma}
\label{second-randomisation}Under Assumptions \ref{ass-1}-\ref{ass-4}, it
holds that
\begin{equation}
\lim \sup_{N,m,R\rightarrow \infty }\frac{\Theta _{t}}{\widetilde{l}\left(
N,m,R\right) }=0\text{ \ \ a.s.,} 
\begin{array}{ll}
\text{under }H_{A,1}, & \text{for }\ t_{N,m}^{\ast }\leq t<t_{N,m}^{\ast\ast }, \\ 
\text{under }H_{A,2}, & \text{for }\ t_{N,m}^{\ast }\leq t\leq T,%
\end{array}
 \label{secondrand-1}
\end{equation}
and
%
\begin{equation}
\frac{\widetilde{l}\left( N,m,R\right) }{R}\times \frac{\Theta _{t}}{%
\widetilde{l}\left( N,m,R\right) }\rightarrow C>0\text{ \ \ a.s.,}\begin{array}{ll}
\text{under }H_{0}, & \text{for }\ m\leq t\leq T, \\ 
\text{under }H_{A,1}, & \text{for }\ m\leq t<\tau, \\ 
& \text{and } \tau+m-1\leq t\leq T, \\ 
\text{under }H_{A,2}, & \text{for }\ m\leq t<\tau ,%
\end{array}
\label{second-rand-2}
\end{equation}%

as $\min \left( N,m,R\right) \rightarrow \infty $. 

\begin{proof}
We begin with (\ref{secondrand-1}). For any $1\leq n\leq N$, $1\leq s\leq m$, $1\leq r\leq
R$, define
\begin{equation*}
U_{n,s,r}=\int_{-\infty }^{+\infty }\bigg\vert r^{-1/2}\sum_{j=1}^{r}\left\{
I\left[ \xi _{j}\leq u\phi _{n,s}^{-1}\left( t\right) \right] -G_{\phi
}\left( 0\right) \right\} \bigg\vert ^{2}dF_{\phi }\left( u\right) .
\end{equation*}%

We begin by showing that%
\begin{equation}
\sum_{N=1}^{\infty }\sum_{m=1}^{\infty }\sum_{R=1}^{\infty }\frac{1}{mNR}%
P^{\ast }\left[ \max_{1\leq n\leq N,1\leq s\leq m,1\leq r\leq
R}U_{n,s,r}>\epsilon \widetilde{l}\left( N,m,R\right) \right] <\infty ,
\label{bc-1}
\end{equation}%
for any $\epsilon>0$. Using the short-hand notation $\max_{n,s,r}$\ for $\max_{1\leq n\leq N,1\leq
s\leq m,1\leq r\leq R}$, Markov inequality implies that \eqref{bc-1} follows if%
\begin{equation}
\sum_{N=1}^{\infty }\sum_{m=1}^{\infty }\sum_{R=1}^{\infty }\frac{1}{mNR%
\widetilde{l}\left( N,m,R\right) }E^{\ast }\left\vert
\max_{n,s,r}U_{n,s,r}\right\vert <\infty .  \label{bc-2}
\end{equation}%
The maximal inequality contained in Theorem 2 in \citet{moricz1983} entails
that 
\begin{equation}
E^{\ast }\left\vert \max_{n,s,r}U_{n,s,r}\right\vert \leq C_{0}E^{\ast
}\left\vert U_{N,m,R}\right\vert \left( \ln N\right) \left( \ln m\right)
\left( \ln R\right) .  \label{moricz}
\end{equation}%
Further, combining (\ref{th1-1}) with (\ref{th1-2})-(\ref{th1-3}), it is easy to see that %
\begin{equation*}
E^{\ast }\left\vert U_{N,m,R}\right\vert \leq C_{0}+C_{1}R\phi
_{N,m}^{-2}\left( t\right),
\end{equation*}%
which holds under $H_{A,1}$ for $t_{N,m}^{\ast}\le t<t_{N,m}^{\ast\ast}$ and under $H_{A,2}$ for $t_{N,m}^{\ast}\le t\le T$. By Assumption \ref{restriction-1}, it holds that $E^{\ast }\left\vert U_{N,m,R}\right\vert$ is bounded. Then (\ref{bc-1})
follows immediately from (\ref{bc-2}).

Note now that for every triple $\left( N,m,R\right) $, there is a triple of
positive integers $\left( k_{1},k_{2},k_{3}\right) $ such that $2^{k_{1}}$ $%
\leq $ $N$ $<$ $2^{k_{1}+1}$, $2^{k_{2}}$ $\leq $ $m$ $<$ $2^{k_{2}+1}$, $%
2^{k_{3}}$ $\leq $ $R$ $<$ $2^{k_{3}+1}$. Further, there is also a triple of
real numbers defined over $\left[ 0,1\right) $, say $\left( \rho _{1},\rho
_{2},\rho _{3}\right) $, such that $N$ $=$ $2^{k_{1}+\rho _{1}}$, etc...
Define now the short-hand notation
\begin{equation*}
A_{k_{1},k_{2},k_{3}}=\left\{ \omega :\max_{1\leq k_{1}\leq 2^{k_{1}+\rho
_{1}},1\leq k_{2}\leq 2^{k_{2}+\rho _{2}},1\leq k_{3}\leq 2^{k_{3}+\rho
_{3}}}\left\vert U_{k_{1},k_{2},k_{3}}\right\vert >\epsilon \widetilde{l}%
\left( k_{1},k_{2},k_{3}\right) \right\} .
\end{equation*}%
By (\ref{bc-1}), it holds that%
\begin{equation*}
\sum_{k_{1}=0}^{\infty }\sum_{k_{2}=0}^{\infty }\sum_{k_{3}=0}^{\infty }%
\frac{2^{k_{1}+1}2^{k_{2}+1}2^{k_{3}+1}}{\left( 2^{k_{1}+1}-1\right) \left(
2^{k_{2}+1}-1\right) \left( 2^{k_{3}+1}-1\right) }P^{\ast }\left(
A_{k_{1},k_{2},k_{3}}\right) <\infty ;
\end{equation*}%
thus%
\begin{eqnarray}
&&\sum_{k_{1}=0}^{\infty }\sum_{k_{2}=0}^{\infty }\sum_{k_{3}=0}^{\infty
}P^{\ast }\left( A_{k_{1},k_{2},k_{3}}\right)\leq   \label{bc-3} \\
&\leq &2^{3}\sum_{k_{1}=0}^{\infty }\sum_{k_{2}=0}^{\infty
}\sum_{k_{3}=0}^{\infty }\frac{2^{k_{1}}2^{k_{2}}2^{k_{3}}}{\left(
2^{k_{1}+1}-1\right) \left( 2^{k_{2}+1}-1\right) \left( 2^{k_{3}+1}-1\right) 
}P^{\ast }\left( A_{k_{1},k_{2},k_{3}}\right) <\infty .  \notag
\end{eqnarray}%
This result entails that $P^{\ast }\left( A_{k_{1},k_{2},k_{3}}\text{ i.o.}%
\right) =1$, which is a \textit{conditional} result. Let now $%
\mathrm X_{k_{1},k_{2},k_{3}}$ be the indicator of $A_{k_{1},k_{2},k_{3}}$, and
note that $A_{k_{1},k_{2},k_{3}}$ is conditional on the $\sigma $-field $%
\mathcal{F}_{k_{1},k_{2},k_{3}}$ $=$ $\left\{ X_{i,t},\text{ }1\leq i\leq
N,1\leq t\leq m\right\} \cup \left\{ \xi _{j},1\leq j\leq R\right\} $, which
is non decreasing. Equation (\ref{bc-3}) implies that%
\begin{equation*}
\sum_{k_{1}=0}^{\infty }\sum_{k_{2}=0}^{\infty }\sum_{k_{3}=0}^{\infty
}E\left( \left. \mathrm X_{k_{1},k_{2},k_{3}}\right\vert \mathcal{F}%
_{k_{1},k_{2},k_{3}}\right) <\infty ;
\end{equation*}%
hence, by Theorem 1 in \citet{chen1978}, it holds that 
\begin{equation}
\sum_{k_{1}=0}^{\infty }\sum_{k_{2}=0}^{\infty }\sum_{k_{3}=0}^{\infty
}\mathrm X_{k_{1},k_{2},k_{3}}<\infty \text{ \ \ a.s.}  \label{bc-5}
\end{equation}%
We note that the result by \citet{chen1978} is for a series indexed by a single
index, but his arguments can be readily generalised to the case of
multi-index series. Equation (\ref{bc-5}) can be equivalently rewritten as 
\begin{equation}
\sum_{k_{1}=0}^{\infty }\sum_{k_{2}=0}^{\infty }\sum_{k_{3}=0}^{\infty
}P\left( A_{k_{1},k_{2},k_{3}}\right) <\infty ,  \label{bc-4}
\end{equation}%
which is an \textit{unconditional} result. From (\ref{bc-4}), it is easy to
see that%
\begin{equation*}
\frac{\max_{k_{1},k_{2},k_{3}}\left\vert U_{k_{1},k_{2},k_{3}}\right\vert }{%
\widetilde{l}\left( k_{1},k_{2},k_{3}\right) }\rightarrow 0\text{ a.s.;}
\end{equation*}%
this entails that 
\begin{equation*}
\frac{\left\vert U_{N,m,R}\right\vert }{\widetilde{l}\left( N,m,R\right) }%
\leq \frac{\max_{k_{1},k_{2},k_{3}}\left\vert
U_{k_{1},k_{2},k_{3}}\right\vert }{\widetilde{l}\left(
k_{1},k_{2},k_{3}\right) }\frac{\widetilde{l}\left( k_{1},k_{2},k_{3}\right) 
}{\widetilde{l}\left( N,m,R\right) }\leq \frac{\max_{k_{1},k_{2},k_{3}}\left\vert U_{k_{1},k_{2},k_{3}}\right\vert }{\widetilde{l}\left(
k_{1},k_{2},k_{3}\right) }\rightarrow 0\text{ a.s.,}
\end{equation*}%
so that finally%
\begin{equation*}
\lim \sup_{N,m,R\rightarrow \infty }\frac{\left\vert U_{N,m,R}\right\vert }{%
\widetilde{l}\left( N,m,R\right) }=0\text{ a.s.,}
\end{equation*}%
from which (\ref{secondrand-1}) follows. 

Consider now (\ref{second-rand-2}). Under $H_0$, Lemma \ref%
{rolling-eigenvalues} entails that 
\begin{equation*}
P\left\{ \omega :\;\lim_{N,m\rightarrow \infty }\phi _{N,m}\left( t\right)
=0\right\} =1,\quad m\leq t\leq
T, \label{phi-nt-1-1}
\end{equation*}%
so that we can assume henceforth that $\lim_{N,m\rightarrow \infty }\phi
_{N,m}\left( t\right) =0$ for $m\leq t\leq
T$. Similarly this holds also under $H_{A,1}$ for $m\le t<\tau$ and $\tau+m-1\le t\le T$ and under $H_{A,2}$ for $m\le t<\tau$.
Also, by definition it holds that $E^{\ast }I%
\left[ \xi _{j}\leq u\phi _{N,m}^{-1}\left( t\right) \right] =G_{\phi
}\left( u\phi _{N,m}^{-1}\left( t\right) \right) $. Therefore 
\begin{align*}
&G_{\phi }\left( 0\right) \left[ 1-G_{\phi }\left( 0\right) \right] \Theta
_{t} =\\
=\,&\int_{-\infty }^{+\infty }\left\vert R^{-1/2}\sum_{j=1}^{R}\left\{ I%
\left[ \xi _{j}\leq u\phi _{N,m}^{-1}\left( t\right) \right] -G_{\phi
}\left( u\phi _{N,m}^{-1}\left( t\right) \right)+G_{\phi
}\left( u\phi _{N,m}^{-1}\left( t\right) \right) -G_{\phi }\left( 0\right)
\right\} \right\vert ^{2}dF_{\phi }\left( u\right) = \\
=&\,\int_{-\infty }^{+\infty }\left\vert R^{-1/2}\sum_{j=1}^{R}\left\{ I\left[
\xi _{j}\leq u\phi _{N,m}^{-1}\left( t\right) \right] -G_{\phi }\left( u\phi
_{N,m}^{-1}\left( t\right) \right) \right\} \right\vert ^{2}dF_{\phi }\left(
u\right) + \\
&+R^{1/2}\int_{-\infty }^{+\infty }\left\vert G_{\phi }\left( u\phi
_{N,m}^{-1}\left( t\right) \right) -G_{\phi }\left( 0\right) \right\vert
^{2}dF_{\phi }\left( u\right) + \\
&+2\int_{-\infty }^{+\infty }R^{-1/2}\sum_{j=1}^{R}\left\{ I\left[ \xi
_{j}\leq u\phi _{N,m}^{-1}\left( t\right) \right] -G_{\phi }\left( u\phi
_{N,m}^{-1}\left( t\right) \right) \right\} \Big[ G_{\phi }\left( u\phi
_{N,m}^{-1}\left( t\right) \right) -G_{\phi }\left( 0\right) \Big]
dF_{\phi }\left( u\right) .
\end{align*}%
Note that, by (\ref{phi-nt-1-1}), $G_{\phi }\left( u\phi _{N,m}^{-1}\left( t\right) \right)
-G_{\phi }\left( 0\right) =I_{\left[ 0,\infty \right) }\left( u\right)
-G_{\phi }\left( 0\right) $ as $N,m\rightarrow \infty $. Also, using similar arguments as in the proof of (\ref{secondrand-1}), it is easy to see that 
\begin{equation}
\lim\!\!\!\!\!\!\! \sup_{N,m,R\rightarrow \infty }\frac{\int_{-\infty }^{+\infty
}\left\vert R^{-1/2}\sum_{j=1}^{R}\left\{ I\left[ \xi _{j}\leq u\phi
_{N,m}^{-1}\left( t\right) \right] -G_{\phi }\left( u\phi _{N,m}^{-1}\left(
t\right) \right) \right\} \right\vert ^{2}dF_{\phi }\left( u\right) }{%
\widetilde{l}\left( N,m,R\right) }=0\text{ a.s.;}. \label{limsup-lemma-a2}
\end{equation}%
Equation (\ref{second-rand-2}) follows directly from (\ref{phi-nt-1-1}), (\ref{limsup-lemma-a2}) (and the Cauchy-Schwartz inequality), and Assumption \ref%
{ass-4}\textit{(iii)}. We point out that the passages above follow closely to the proof of Theorem 4 in \citet{HT16}. 
\end{proof}
\end{lemma}

\begin{lemma}
\label{gamma-chisq}Under Assumptions \ref{ass-1}-\ref{ass-5} and \ref%
{restriction-3}(ii), it holds that, under $H_{0}$%
\begin{eqnarray}
\max_{1\leq k\leq T_{m}}\sqrt{\frac{m}{k\left( k+m\right) }}\left\vert
\sum_{t=m+1}^{m+k}\left( E^{\dag }\left( \Gamma _{t}\right) -1\right)
\right\vert &=&O\left( m^{-\epsilon }\right) ,  \label{exp-gamma} \\
\max_{1\leq k\leq T_{m}}\sqrt{\frac{m}{k\left( k+m\right) }}\left\vert
\sum_{t=m+1}^{m+k}\left( V^{\dag }\left( \Gamma _{t}\right) -2\right)
\right\vert &=&O\left( m^{-\epsilon }\right) ,  \label{var-gamma}
\end{eqnarray}%
for some $\epsilon >0$. Also%
\begin{equation}
E^{\dag }\left\vert \Gamma _{t}\right\vert ^{2+\delta }<\infty ,
\label{liapunov-gamma}
\end{equation}%
for some $\delta >0$.

\begin{proof}
We start with equation (\ref{exp-gamma}). By construction%
\begin{align*}
&\frac{E^{\dag }\left( \Gamma _{t}\right) }{G_{\psi }\left( 0\right) \left[
1-G_{\psi }\left( 0\right) \right] } =\frac{E^{\dag }\int_{-\infty
}^{+\infty }\left\vert \gamma \left( u;t\right) \right\vert ^{2}dF_{\psi
}\left( u\right) }{G_{\psi }\left( 0\right) \left[ 1-G_{\psi }\left(
0\right) \right] }= \\
&=E^{\dag }\int_{-\infty }^{+\infty }\left\vert W^{-1/2}\sum_{j=1}^{W}\left[
\widetilde{\zeta }_{j}\left( u;t\right) -G_{\psi }\left( 0\right) \right]
\right\vert ^{2}dF_{\psi }\left( u\right) =\\
&=W^{-1}\int_{-\infty }^{+\infty }E^{\dag }\left\vert \widetilde{\zeta }%
_{j}\left( u;t\right) -G_{\psi }\left( 0\right) \right\vert ^{2}dF_{\psi
}\left( u\right) ;
\end{align*}%
by similar passages as in the proof of Theorem \ref{theta}, it can be shown
that%
\begin{equation*}
\frac{E^{\dag }\left( \Gamma _{t}\right) }{G_{\psi }\left( 0\right) \left[
1-G_{\psi }\left( 0\right) \right] }-1\leq C_{0}\left[ Wh^{-2}\left( \frac{R%
}{\ln R}\right) +h^{-1}\left( \frac{R}{\ln R}\right) \right] .
\end{equation*}%
Thus%
\begin{align*}
&\max_{1\leq k\leq T_{m}}\sqrt{\frac{m}{k\left( k+m\right) }}\left\vert
\sum_{t=m+1}^{m+k}\left( E^{\dag }\left( \Gamma _{t}\right) -1\right)
\right\vert \leq\\
&\leq C_{0}W\left( \frac{\ln R}{R}\right) ^{2}\max_{1\leq k\leq T_{m}}\sqrt{%
\frac{m}{k\left( k+m\right) }}k\, \leq  \\
&\leq C_{0}m^{1/2}\left[ Wh^{-2}\left( \frac{R}{\ln R}\right) +h^{-1}\left( 
\frac{R}{\ln R}\right) \right] ,
\end{align*}%
which is $O\left( m^{-\epsilon }\right) $ on account of Assumption \ref%
{restriction-3}\textit{(ii)}.

We now turn to (\ref{var-gamma}). Let $\gamma \left( 0;t\right)
=W^{-1/2}\sum_{j=1}^{W}\left[ \widetilde{\zeta }_{j}\left( 0;t\right)
-G_{\psi }\left( 0\right) \right] $; we have%
\begin{align}
&\left( \int_{-\infty }^{+\infty }\left\vert \gamma \left( u;t\right)
\right\vert ^{2}dF_{\psi }\left( u\right) \right) ^{2}-\left( \int_{-\infty
}^{+\infty }\left\vert \gamma \left( 0;t\right) \right\vert ^{2}dF_{\psi
}\left( u\right) \right) ^{2} = \label{var-diff} \\
&=\int_{-\infty }^{+\infty }\left( \left\vert W^{-1/2}\sum_{j=1}^{W}\left[ 
\widetilde{\zeta }_{j}\left( u;t\right) -\widetilde{\zeta }_{j}\left(
0;t\right) \right] \right\vert ^{2}\right. + \notag \\
&\left. +2W^{-1}\sum_{j=1}^{W}\left[ \widetilde{\zeta }_{j}\left(
u;t\right) -\widetilde{\zeta }_{j}\left( 0;t\right) \right] \sum_{j=1}^{W}%
\left[ \widetilde{\zeta }_{j}\left( 0;t\right) -G_{\psi }\left( 0\right) %
\right] ^{2}\right) ^{2}dF_{\psi }\left( u\right)  +\notag \\
&+2\int_{-\infty }^{+\infty }W^{-1/2}\sum_{j=1}^{W}\left[ \widetilde{\zeta }%
_{j}\left( 0;t\right) -G_{\psi }\left( 0\right) \right] \times \left(
\left\vert W^{-1/2}\sum_{j=1}^{W}\left[ \widetilde{\zeta }_{j}\left(
u;t\right) -\widetilde{\zeta }_{j}\left( 0;t\right) \right] \right\vert
^{2}\right. +  \notag \\
&+\left. 2W^{-1}\sum_{j=1}^{W}\left[ \widetilde{\zeta }_{j}\left( u;t\right)
-\widetilde{\zeta }_{j}\left( 0;t\right) \right] \sum_{j=1}^{W}\left[ 
\widetilde{\zeta }_{j}\left( 0;t\right) -G_{\psi }\left( 0\right) \right]
\right) dF_{\psi }\left( u\right) .  \notag
\end{align}%
By Rosenthal's inequality%
\begin{align}
&W^{-2}E^{\dag }\left\vert \sum_{j=1}^{W}\left[ \widetilde{\zeta }%
_{j}\left( u;t\right) -\widetilde{\zeta }_{j}\left( 0;t\right) \right]
\right\vert ^{4}\leq  \label{rosenthal} \\
&\leq C_{0}W^{-2}\left[ \left\vert \sum_{j=1}^{W}E^{\dag }\left( \widetilde{%
\zeta }_{j}\left( u;t\right) -\widetilde{\zeta }_{j}\left( 0;t\right)
\right) \right\vert ^{4}+\sum_{j=1}^{W}E^{\dag }\left\vert \widetilde{\zeta }%
_{j}\left( u;t\right) -\widetilde{\zeta }_{j}\left( 0;t\right) \right\vert
^{4}\right] \leq \notag \\
&\leq C_{0}W^{-2}\left[ W^{4}\left\vert E^{\dag }\left( \widetilde{\zeta }%
_{j}\left( u;t\right) -\widetilde{\zeta }_{j}\left( 0;t\right) \right)
\right\vert ^{4}+C_{1}W\right] \leq \notag \\
&\leq C_{0}W^{-2}\left[ W^{4}\left\vert G_{\psi }\left( \psi
_{N,m,R}^{-1}\left( t\right) \right) -G_{\psi }\left( 0\right) \right\vert
^{4}+C_{1}W\right] \leq \notag \\
&\leq C_{0}W^{-1}+C_{1}W^{2}u^{4}\psi _{N,m,R}^{-4}\left( t\right) ;  \notag
\end{align}%
the same logic also yields%
\begin{align*}
&W^{-1}E^{\dag }\left\vert \sum_{j=1}^{W}\left[ \widetilde{\zeta }_{j}\left(
u;t\right) -\widetilde{\zeta }_{j}\left( 0;t\right) \right] \right\vert ^{2}
\leq C_{0}u\psi _{N,m,R}^{-1}\left( t\right) +C_{1}Wu^{2}\psi
_{N,m,R}^{-2}\left( t\right) , \\
&E^{\dag }\left\vert W^{-1/2}\sum_{j=1}^{W}\left[ \widetilde{\zeta }%
_{j}\left( 0;t\right) -G_{\psi }\left( 0\right) \right] \right\vert ^{2}
\leq C_{0}.
\end{align*}%
Repeated applications of the $C_{r}$-inequality and of the Cauchy-Schwartz
inequality to (\ref{var-diff}) yield%
\begin{equation*}
E^{\dag }\left( \int_{-\infty }^{+\infty }\left\vert \gamma \left(
u;t\right) \right\vert ^{2}dF_{\psi }\left( u\right) \right) ^{2}-E^{\dag
}\left( \int_{-\infty }^{+\infty }\left\vert \gamma \left( 0;t\right)
\right\vert ^{2}dF_{\psi }\left( u\right) \right) ^{2}\leq C_{0}W\psi
_{N,m,R}^{-2}\left( t\right) .
\end{equation*}%
\ Finally, tedious but elementary calculations yield%
\begin{equation*}
E^{\dag }\left\vert W^{-1/2}\sum_{j=1}^{W}\left[ \widetilde{\zeta }%
_{j}\left( 0;t\right) -G_{\psi }\left( 0\right) \right] \right\vert
^{4}=3\left( G_{\psi }\left( 0\right) \left[ 1-G_{\psi }\left( 0\right) %
\right] \right) ^{2}+O\left( W^{-1}\right) .
\end{equation*}%
Putting all together and using (\ref{exp-gamma}), it follows that%
\begin{equation*}
V^{\dag }\left( \Gamma _{t}\right) -2\leq C_{0}\left( W\psi
_{N,m,R}^{-2}\left( t\right) +W^{-1}\right) ,
\end{equation*}%
whence (\ref{var-gamma}) follows.

Consider now (\ref{liapunov-gamma}). By convexity%
\begin{equation*}
\left\vert \int_{-\infty }^{+\infty }\left\vert \gamma \left( u;t\right)
\right\vert ^{2}dF_{\psi }\left( u\right) \right\vert ^{2+\delta }\leq
\int_{-\infty }^{+\infty }\left\vert \gamma \left( u;t\right) \right\vert
^{4+2\delta }dF_{\psi }\left( u\right) ;
\end{equation*}%
also, applying the $C_{r}$-inequality%
\begin{align*}
&\int_{-\infty }^{+\infty }\left\vert \gamma \left( u;t\right) \right\vert
^{4+2\delta }dF_{\psi }\left( u\right) \leq \\
&\leq C_{0}\int_{-\infty }^{+\infty }\left\vert W^{-1/2}\sum_{j=1}^{W}\left[
\widetilde{\zeta _{j}}\left( u;t\right) -G_{\psi }\left( u\psi
_{N,m,R}^{-1}\left( t\right) \right) \right] \right\vert ^{4+2\delta
}dF_{\psi }\left( u\right) +\\
&+C_{0}\int_{-\infty }^{+\infty }\left\vert W^{-1/2}\sum_{j=1}^{W}\left[
G_{\psi }\left( u\psi _{N,m,R}^{-1}\left( t\right) \right) -G_{\psi }\left(
0\right) \right] \right\vert ^{4+2\delta }dF_{\psi }\left( u\right) .
\end{align*}%
Note that, by applying Burkholder's inequality and convexity%
\begin{align}
&\int_{-\infty }^{+\infty }E^{\dag }\left\vert W^{-1/2}\sum_{j=1}^{W}\left[ 
\widetilde{\zeta _{j}}\left( u;t\right) -G_{\psi }\left( u\psi
_{N,m,R}^{-1}\left( t\right) \right) \right] \right\vert ^{4+2\delta
}dF_{\psi }\left( u\right)\leq  \nn\\
&\leq \int_{-\infty }^{+\infty }E^{\dag }\left\vert W^{-1}\sum_{j=1}^{W} 
\left[ \widetilde{\zeta _{j}}\left( u;t\right) -G_{\psi }\left( u\psi
_{N,m,R}^{-1}\left( t\right) \right) \right] ^{2}\right\vert ^{2+\delta
}dF_{\psi }\left( u\right)\leq  \nn\\
&\leq W^{-1}\sum_{j=1}^{W}\int_{-\infty }^{+\infty }E^{\dag }\left\vert 
\left[ \widetilde{\zeta _{j}}\left( u;t\right) -G_{\psi }\left( u\psi
_{N,m,R}^{-1}\left( t\right) \right) \right] ^{2}\right\vert ^{2+\delta
}dF_{\psi }\left( u\right) \leq C_{0},\nn
\end{align}%
since $\widetilde{\zeta _{j}}\left( u;t\right) $ has finite moments of any
order. Also%
\begin{align}
&\int_{-\infty }^{+\infty }\left\vert W^{-1/2}\sum_{j=1}^{W}\left[ G_{\psi
}\left( u\psi _{N,m,R}^{-1}\left( t\right) \right) -G_{\psi }\left( 0\right) %
\right] \right\vert ^{4+2\delta }dF_{\psi }\left( u\right)\leq\nn \\
&\leq W^{2+\delta }\int_{-\infty }^{+\infty }\left\vert G_{\psi }\left(
u\psi _{N,m,R}^{-1}\left( t\right) \right) -G_{\psi }\left( 0\right)
\right\vert ^{4+2\delta }dF_{\psi }\left( u\right)\nn\leq \\
&\leq W^{2+\delta }\left( \frac{m_{G,\psi }}{\psi _{N,m,R}\left( t\right) }%
\right) ^{4+2\delta }\int_{-\infty }^{+\infty }\left\vert u\right\vert
^{4+2\delta }dF_{\psi }\left( u\right) \leq C_{0},\nn
\end{align}%
where $m_{G,\psi }$ is the upper bound of the density of $G_{\psi }\left(
\cdot \right) $; the final estimate follows from Assumptions \ref%
{restriction-2} and \ref{restriction-3}\textit{(i)}. This proves (\ref%
{liapunov-gamma}).
\end{proof}
\end{lemma}

\subsection{Proofs of main results}

\begin{proof}[Proof of Lemma \protect\ref{rolling-eigenvalues}]

Recall the definition%
\begin{equation*}
\Sigma _{m}\left( t\right) =\frac{1}{m}\sum_{k=t-m+1}^{t}E\left(
X_{k}X_{k}^{\prime }\right) .
\end{equation*}%
When $m\le t<\tau $ and under $H_{A,1}$ or $H_{A,2}$, it is easy to see that there is no change in the $\left(
r+1\right) $-th eigenvalue of $E\left( X_{k}X_{k}^{\prime }\right) $. Thus,
the proof that $\lambda ^{\left( r+1\right) }\left( t\right) $ is finite is
exactly the same as the proof of Lemma 2.1 in \citet{trapani17}. The same holds under $H_0$ for all $m\le t\le T$.

We begin with studying $H_{A,2}$. When $\tau\le t<\tau+m-1$, it holds that%
\begin{align}
\Sigma _{m}\left( t\right) &=\frac{1}{m}\sum_{k=t-m+1}^{\tau -1}E\left(
X_{k}X_{k}^{\prime }\right) +\frac{1}{m}\sum_{k=\tau }^{t}E\left(
X_{k}X_{k}^{\prime }\right)\nonumber\\
& =\frac{\tau+m-t-1 }{m}\Sigma _{m}^{\left(
1\right) }\left( t\right) +\frac{t-\tau +1}{m}\Sigma _{m}^{\left( 2\right)
}\left( t\right) .  \label{sig-m-t}
\end{align}

Let $\lambda _{1}^{\left( r+1\right) }\left( t\right) $ and $\lambda
_{2}^{\left( r+1\right) }\left( t\right) $ be the $\left( r+1\right) $-th
eigenvalue of $\Sigma _{m}^{\left( 1\right) }\left( t\right) $ and $\Sigma
_{m}^{\left( 2\right) }\left( t\right) $ respectively. Now $\lambda _{1}^{\left( r+1\right) }\left( t\right)$ depends only on observations before the break, therefore $\lambda _{1}^{\left( r+1\right) }\left( t\right)<\infty$: this
can be shown again by following the proof of Lemma 2.1 in \citet{trapani17}. As far as 
$\lambda _{2}^{\left( r+1\right) }\left( t\right) $ is concerned, it depends only on 
post-break observations, which are driven by a factors vector $f_t$ of size $(r+q)$. In particular,
$\Sigma _{m}^{\left( 2\right) }\left( t\right) =A(t)E\left( f_{t}f_{t}^{\prime }\right) A(t)^{\prime }+\Sigma
_{u}\left( t\right) $, where $A(t)=\left[ A | B \right]$ is a constant $N\times (r+q)$ matrix.
By Weyl's inequality and Assumption \ref{ass-2} 

\begin{align}
\lambda _{2}^{\left( r+1\right) }\left( t\right) &
\ge \gamma^{(r+1)}(t) +\omega^{(N)}(t)\geq \underline C_{r+1}(t)N.  \label{sig-m2-t}
\end{align}

Then, applying Weyl's inequality
and (\ref{sig-m2-t}) to (\ref{sig-m-t}), it follows that%
\begin{equation*}
\lambda ^{\left( r+1\right) }\left( t\right) \geq \frac{\tau+m-t-1 }{m}%
\lambda _{1}^{\left( \min \right) }\left( t\right) +\frac{t-\tau +1}{m}%
\lambda _{2}^{\left( r+1\right) }\left( t\right) ,
\end{equation*}%
which yields (\ref{lambda-population-2}) immediately. When instead $\tau+m-1\le t \le T$ we have $\Sigma_m(t)=\Sigma_m^{(2)}(t)$ and the result follows directly from  \eqref{sig-m2-t}. 
Under $H_{A,1}$, $r(t)=r$ for $m\le t\le T$ and therefore $\Sigma_F(t)= E(f_tf_t^\prime)$ is a $r\times r$ constant matrix and we denote it as $\Sigma_F$.
We have
\beq\label{eq:boh}
\Sigma_X(t) = \big(A\Sigma_F A'\big) \ I_{[m,\tau)}(t) + \big(\widetilde A\Sigma_F\widetilde A' \big)\, I_{[\tau,T]}(t)+\Sigma_u(t). 
\eeq
Therefore, it holds that for $\tau\le t< \tau+m-1$
\begin{equation*}
\Sigma _{m}\left( t\right) =\frac 1 m\sum_{k=t-m+1}^{t}\Sigma_X(k) =A^{\ast }\Sigma _{F}^{\ast }\left( t\right)
A^{\ast \prime }+\Sigma _{u}\left( t\right) ,
\end{equation*}%
where we have defined $A^*=[A|\widetilde A]$ and
\begin{equation}
\Sigma _{F}^{\ast }\left( t\right) =\left[ 
\begin{array}{cc}
\frac{\tau+m-t-1 }{m}\Sigma _{F}  & 0 \\ 
0 & \frac{t-\tau +1}{m}\Sigma _{F}
\end{array}%
\right] ,  \label{sig-f-star}
\end{equation}%
with the off-diagonal blocks being $r\times r$ matrices of zeros. Denoting by $\nu^{(k)}(\cdot)$ the $k$-th largest eigenvalues of a matrix, by Weyl's
inequality, we have
\begin{equation}
\lambda ^{\left( r+1\right) }\left( t\right) \geq \nu ^{\left(
r+1\right) }\left( A^{\ast }\Sigma _{F}^{\ast }\left( t\right) A^{\ast
\prime }\right) +\omega^{(N)}(t)\ge  \nu ^{\left(
r+1\right) }\left( A^{\ast }\Sigma _{F}^{\ast }\left( t\right) A^{\ast
\prime }\right),
\label{weyl-m-Ha2}
\end{equation}
Now, since the spectrum of a block diagonal matrix is the union of the spectra of the blocks 
\begin{align}
\nu^{\left( r+1\right) }&\left(
A^{\ast}\Sigma_F^{\ast}(t)A^{\ast \prime  }\right)= \nu^{\left(r+1\right)}\left(\begin{array}{cc}
\frac{\tau+m-t-1 }{m}A\Sigma_FA'&0\\
0&\frac{t-\tau +1}{m}\widetilde A\Sigma_F\widetilde A'
\end{array}
\right)\nonumber\\
&= \min\left\{ \nu^{(r)}\left(\frac{\tau+m-t-1 }{m}A\Sigma_F A'\right),\nu^{(r)}\left(\frac{t-\tau +1}{m}\widetilde A\Sigma_F\widetilde A'\right)\right\}\nonumber\\
&\ge \min\left\{\frac{\tau+m-t-1 }{m},\frac{t-\tau +1}{m} \right\}\,\min\left\{\nu^{(r)}\big(A\Sigma_F A'\big),\nu^{(r)}\big(\widetilde A\Sigma_F\widetilde A'\big)\right\}\nonumber\\
&\ge\min\left\{\frac{\tau+m-t-1 }{m},\frac{t-\tau +1}{m} \right\} \underline C_r N,\label{eq:SFmax}
\end{align}
where the last inequality follows from the fact that by \eqref{eq:boh} and Assumption \ref{ass-2} we have
$$
\gamma^{(r)}(t)= \nu^{(r)}\big(A\Sigma_FA'\big)  I_{[m,\tau)}(t)+ \nu^{(r)}\big(\widetilde A\Sigma_F\widetilde A'\big)  I_{[\tau,T]}(t)\ge \underline C_r N,
$$  
which implies that $ \nu^{(r)}\big(A\Sigma_FA'\big)\ge \underline C_r N$ and $\nu^{(r)}\big(\widetilde A\Sigma_F\widetilde A'\big)\ge \underline C_r N$. 
Using \eqref{eq:SFmax} in (\ref{weyl-m-Ha2}), equation (\ref{lambda-population}) follows. Last, when $\tau+m-1\le t\le T$, we have  $A^*=\widetilde A$ and $\Sigma_F^*(t)=\Sigma_F$ and the proof is the same as in Lemma 2.1 in \citet{trapani17}.
\end{proof}

\begin{proof}[Proof of Lemma \protect\ref{rolling-eigenvalues-2}]

The proof is essentially the same as the proof of Lemma 2.2 in \citet{trapani17}, and we report it here in full for completeness. Consider the
eigenvalue stability inequality (see e.g. \citet{hornjohnson}, p. 367), viz.%
\begin{equation}
\left\vert \widehat{\lambda }^{\left( r+1\right) }\left( t\right) -\lambda
^{\left( r+1\right) }\left( t\right) \right\vert \leq \left\Vert \widehat{%
\Sigma }_{m}\left( t\right) -\Sigma _{m}\left( t\right) \right\Vert _{op},
\label{beta1}
\end{equation}%
where $\left\Vert \cdot \right\Vert _{op}$ is the operator norm. By
symmetry, $\left\Vert \widehat{\Sigma }_{m}\left( t\right) -\Sigma
_{m}\left( t\right) \right\Vert _{op}\leq \left\Vert \widehat{\Sigma }%
_{m}\left( t\right) -\Sigma _{m}\left( t\right) \right\Vert _{F}$, where $%
\left\Vert \cdot \right\Vert _{F}$ denotes the Frobenius norm; hence, (\ref%
{beta1}) becomes
\begin{equation}
\left\vert \widehat{\lambda }^{\left( r+1\right) }\left( t\right) -\lambda
^{\left( r+1\right) }\left( t\right) \right\vert \leq \left[
\sum_{i=1}^{N}\sum_{j=1}^{N}\left( \frac{1}{m}%
\sum_{k=t-m+1}^{t}X_{i,k}X_{j,k}-E(X_{i,k}X_{j,k})\right) ^{2}\right] ^{1/2}.
\label{betahigh}
\end{equation}%
We now provide an estimate for $\left\vert \widehat{\lambda }^{\left(
r+1\right) }\left( t\right) -\lambda ^{\left( r+1\right) }\left( t\right)
\right\vert $; the proof uses, in a multi-index context, the same approach
as \citet{cai2006}. Let $\delta _{h,j,k}\equiv X_{h,k}X_{j,k}-E(X_{h,k}X_{j,k})$, and $%
t_{0}=t-m+1$; we begin by showing 
\begin{align} 
\sum_{N=1}^{\infty }\sum_{N=1}^{\infty }\sum_{m=1}^{\infty }\frac{1}{N^{2}m}P
\Bigg[ \max_{1\leq \tilde{h}\leq N,1\leq \tilde{j}\leq N,t_{0}\leq \tilde{t}%
\leq t_{0}+m-1}\Bigg\vert \sum_{h=1}^{\tilde{h}}\sum_{j=1}^{\tilde{j}}\Bigg( 
\frac{1}{m}\sum_{k=t_{0}}^{\tilde{t}}\delta _{h,j,k}\Bigg) ^{2}\Bigg\vert
^{1/2}\nonumber\\
>\varepsilon \frac{N}{\sqrt{m}}\ln ^{1+\epsilon }N\ln ^{\frac{%
1+\epsilon }{2}}m\Bigg]  < \infty ,  \label{pmax1}
\end{align}%
for some $\varepsilon >0$ and any $\epsilon>0$. Equation (\ref{pmax1}) can be shown by noting that%
\begin{align*}
E\left[ \max_{1\leq \tilde{h}\leq N,1\leq \tilde{j}\leq N,t_{0}\leq \tilde{t}%
\leq t_{0}+m-1}\left\vert \sum_{h=1}^{\tilde{h}}\sum_{j=1}^{\tilde{j}}\left( 
\frac{1}{m}\sum_{k=t_{0}}^{\tilde{t}}\delta _{h,j,k}\right) ^{2}\right\vert %
\right] \\
\leq \sum_{h=1}^{N}\sum_{j=1}^{N}E\left[ \max_{t_{0}\leq \tilde{t}%
\leq t_{0}+m-1}\left\vert \left( \frac{1}{m}\sum_{k=t_{0}}^{\tilde{t}}\delta
_{h,j,k}\right) ^{2}\right\vert \right] \leq C_{0}\frac{N^{2}}{m},
\end{align*}%
by virtue of Assumption \ref{ass-3}\textit{(ii)}. Thus, by Markov inequality, (\ref%
{pmax1}) holds since 
\begin{equation*}
\sum_{N=1}^{\infty }\sum_{N=1}^{\infty }\sum_{m=1}^{\infty }\frac{1}{N^{2}m}%
\frac{m}{\varepsilon ^{2}N^{2}\left[ \ln ^{1+\epsilon }N\ln ^{\frac{%
1+\epsilon }{2}}m\right] ^{2}}C_{0}\frac{N^{2}}{m}<\infty .
\end{equation*}%
We now show that (\ref{pmax1}) entails 
\begin{equation}
\lim \sup_{N,m\rightarrow \infty }\frac{\left\vert
\sum_{h=1}^{N}\sum_{j=1}^{N}\left(\sum_{k=t_{0}}^{t_{0}+m-1} \delta
_{h,j,k}\right)^{2}\right\vert^{1/2} }{N\sqrt{m}\ln ^{1+\epsilon }N\ln ^{\frac{1+\epsilon }{2}%
}m}=0\text{ a.s.}  \label{pmax2}
\end{equation}%
Similarly to the proof of Lemma \ref{second-randomisation}, note that for every triple $\left( N,N,m\right) $, there is a triple of positive
integers $\left( k'_{1},k'_{2},k'_{3}\right) $ such that $2^{k'_{1}}$ $\leq $ $N$
$<$ $2^{k'_{1}+1}$, $2^{k'_{2}}$ $\leq $ $N$ $<$ $2^{k'_{2}+1}$, $2^{k'_{3}}$ $%
\leq $ $m$ $<$ $2^{k'_{3}+1}$. Further, there is also a triple of real
numbers defined over $\left[ 0,1\right) $, say $\left( \rho' _{1},\rho'
_{2},\rho' _{3}\right) $, such that $N$ $=$ $2^{k'_{1}+\rho' _{1}}$, etc...
Define now the short-hand notation%
\begin{align}
L\left( k'_{1},k'_{2},k'_{3}\right) &\equiv \frac{\sqrt{2^{k'_{1}+1}}\sqrt{%
2^{k'_{2}+1}}}{\sqrt{2^{k'_{3}+1}}}\ln ^{\frac{1+\epsilon }{2}}\left(
2^{k'_{1}+\rho' _{1}}\right) \ln ^{\frac{1+\epsilon }{2}}\left( 2^{k'_{2}+\rho'
_{2}}\right) \ln ^{\frac{1+\epsilon }{2}}\left( 2^{k'_{3}+\rho' _{3}}\right) ,\nonumber\\
S\left( k'_{1},k'_{2},k'_{3}\right) &\equiv
\left\vert\sum_{h=1}^{k'_{1}}\sum_{j=1}^{k'_{2}}\left(\sum_{k=t_{0}}^{t_{0}+k'_{3}-1}\delta
_{h,j,k}\right)^{2}\right\vert^{1/2},\nonumber\\
P_{k'_{1},k'_{2},k'_{3}}&\equiv P\left[ \max_{1\leq k'_{1}\leq 2^{k'_{1}+\rho'
_{1}},1\leq k'_{2}\leq 2^{k'_{2}+\rho' _{2}},1\leq k'_{3}\leq 2^{k'_{3}+\rho'
_{3}}}\left\vert S\left( k'_{1},k'_{2},k'_{3}\right) \right\vert >\varepsilon
L\left( k'_{1},k'_{2},k'_{3}\right) \right] .\nonumber
\end{align}%


Equation (\ref{pmax1}) entails that%
\begin{equation*}
\sum_{k'_{1}=0}^{\infty }\sum_{k'_{2}=0}^{\infty }\sum_{k'_{3}=0}^{\infty }%
\frac{2^{k'_{1}}2^{k'_{2}}2^{k'_{3}}}{\left( 2^{k'_{1}+1}-1\right) \left(
2^{k'_{2}+1}-1\right) \left( 2^{k'_{3}+1}-1\right) }P_{k'_{1},k'_{2},k'_{3}}<%
\infty ;
\end{equation*}%
thus%
\begin{equation*}
\sum_{k'_{1}=0}^{\infty }\sum_{k'_{2}=0}^{\infty }\sum_{k'_{3}=0}^{\infty
}P_{k'_{1},k'_{2},k'_{3}}\leq 2^{3} \sum_{k'_{1}=0}^{\infty }\sum_{k'_{2}=0}^{\infty
}\sum_{k'_{3}=0}^{\infty }\frac{2^{k'_{1}}2^{k'_{2}}2^{k'_{3}}}{\left(
2^{k'_{1}+1}-1\right) \left( 2^{k'_{2}+1}-1\right) \left( 2^{k'_{3}+1}-1\right) 
}P_{k'_{1},k'_{2},k'_{3}}<\infty ,
\end{equation*}%
so that the Borel-Cantelli Lemma yields%
\begin{equation*}
\frac{\max_{k'_{1},k'_{2},k'_{3}}\left\vert S\left( k'_{1},k'_{2},k'_{3}\right)
\right\vert }{L\left( k'_{1},k'_{2},k'_{3}\right) }\rightarrow 0\text{ a.s.,}
\end{equation*}%
whence 
\begin{align}
\frac{\left\vert S\left( N,N,T\right) \right\vert }{L\left( N,N,T\right) }
\leq \frac{\max_{k'_{1},k'_{2},k'_{3}}\left\vert S\left(
k'_{1},k'_{2},k'_{3}\right) \right\vert }{L\left( k'_{1},k'_{2},k'_{3}\right) }%
\frac{L\left( k'_{1},k'_{2},k'_{3}\right) }{L\left( N,N,T\right) } \leq \sqrt{2}\frac{\max_{k'_{1},k'_{2},k'_{3}}\left\vert S\left(
k'_{1},k'_{2},k'_{3}\right) \right\vert }{L\left( k'_{1},k'_{2},k'_{3}\right) }%
\rightarrow 0\text{ a.s.,}\nonumber
\end{align}%
so that finally%
\begin{equation*}
\lim \sup_{N,T\rightarrow \infty }\frac{\left\vert S\left( N,N,T\right)
\right\vert }{L\left( N,N,T\right) }=0\text{ a.s..}
\end{equation*}%
The desired result now follows immediately.

\end{proof}

\begin{proof}[Proof of Theorem \protect\ref{theta}]
Consider (\ref{theta-null}); we report its proof, which is a refinement of the proof of Theorem 3 in \citet{HT16}, in full. In the presence of
a break, it follows from Lemmas \ref{rolling-eigenvalues} and \ref%
{rolling-eigenvalues-2} that
\begin{equation*}
P\left\{ \omega :\;\lim_{N,m\rightarrow \infty }\phi _{N,m}\left( t\right)
=\infty \right\} =1,
\end{equation*}%
for each $t\geq \tau $, as long as%
\begin{eqnarray*}
N^{1-\delta }\min \left\{\frac{t-\tau +1}{m},\frac{\tau+m-t-1}{m}\right \} &\rightarrow &\infty \text{ for }
\tau \leq
t< \tau +m-1, \\
N^{1-\delta } &\rightarrow &\infty \text{ for } t\geq \tau +m-1,
\end{eqnarray*}%
hold. By definition, this holds true within the intervals $t^{\ast}_{N,m} \leq t \leq t^{\ast}_{N,m}$ under $H_{A,1}$, and $t^{*}_{N,m} \leq t \leq T$ under $H_{A,2}$. Thus, we can assume from now on that $\lim_{N,m\rightarrow \infty
}\phi _{N,m}\left( t\right) =\infty $ holds in the prescribed intervals. Note that%
\begin{align}
R^{-1/2}\sum_{i=1}^{R}&\left[ \zeta _{i}(u;t)-G_{\phi }\left( 0\right) \right]
=R^{-1/2}\sum_{i=1}^{R}\left[ I\{\xi _{i}\leq 0\}-G_{\phi }\left( 0\right) %
\right]+ \nonumber \\
& +R^{-1/2}\sum_{i=1}^{R}\left[ I\{\xi _{i}\leq u\phi _{N,m}^{-1}\left( t\right)
\}-I\{\xi _{i}\leq 0\}-(G_{\phi }(u\phi _{N,m}^{-1}\left( t\right) )-G_{\phi
}(0))\right] + \nonumber \\
& +R^{-1/2}\sum_{i=1}^{R}\left[ G_{\phi }(u\phi _{N,m}^{-1}\left( t\right) )-G_{\phi
}(0)\right] . \label{th1-1}
\end{align}

By construction, 
\begin{eqnarray*}
E^{\ast }\zeta _{j}(u;t) &=&G_{\phi }(u\phi _{N,m}^{-1}\left( t\right) ), \\
E^{\ast }(\zeta _{j}(u;t)-E^{\ast }\zeta _{j}(u;t))^{2} &=&G_{\phi }(u\phi
_{N,m}^{-1}\left( t\right) )\left[ 1-G_{\phi }(u\phi _{N,m}^{-1}\left(
t\right) )\right] .
\end{eqnarray*}%
Consider now the following passages:%
\begin{align*}
E^{\ast }& \left[\int_{-\infty }^{\infty }\!\left( R^{-1/2}\sum_{i=1}^{R}\left[
I\{\xi _{i}\leq u\phi _{N,m}^{-1}\left( t\right) \}-I\{\xi _{i}\leq
0\}-(G_{\phi }(u\phi _{N,m}^{-1}\left( t\right) )-G_{\phi }(0))\right]
\right) ^{2}\!dF_{\phi }(u)\right]= \\
& =\int_{-\infty }^{\infty }E^{\ast }\left[ I\{\xi _{1}\leq u\phi
_{N,m}^{-1}\left( t\right) \}-I\{\xi _{1}\leq 0\}-(G_{\phi }(u\phi
_{N,m}^{-1}\left( t\right) )-G_{\phi }(0))\right] ^{2}dF_{\phi }(u),
\end{align*}%
on account of the independence of the $\xi _{i}$s. Also, note that the
random variable $I\{\xi _{1}\leq u\phi _{N,m}^{-1}\left( t\right) \}-I\{\xi
_{1}\leq 0\}$ has expected value given by $G_{\phi }(u\phi _{N,m}^{-1}\left(
t\right) )-G_{\phi }(0)$, and variance equal to%
\begin{align*}
& E^{\ast }\left[ I\{\xi _{1}\leq u\phi _{N,m}^{-1}\left( t\right) \}-I\{\xi
_{1}\leq 0\}-G_{\phi }(u\phi _{N,m}^{-1}\left( t\right) )-G_{\phi }(0)\right]
^{2}= \\
& =\left( G_{\phi }(u\phi _{N,m}^{-1}\left( t\right) )-G_{\phi }(0)\right) 
\left[ 1-G_{\phi }(u\phi _{N,m}^{-1}\left( t\right) )-G_{\phi }(0)\right]\leq \\
& \leq G_{\phi }(u\phi _{N,m}^{-1}\left( t\right) )-G_{\phi }(0).
\end{align*}%
Hence, we have%
\begin{align}
& \int_{-\infty }^{\infty }E^{\ast }\left[ I\{\xi _{1}\leq u\phi
_{N,m}^{-1}\left( t\right) \}-I\{\xi _{1}\leq 0\}-(G_{\phi }(u\phi
_{N,m}^{-1}\left( t\right) )-G_{\phi }(0))\right] ^{2}dF_{\phi }(u) \leq \nonumber \\
& \leq \int_{-\infty }^{\infty }\left[ G_{\phi }(u\phi _{N,m}^{-1}\left(
t\right) )-G_{\phi }(0)\right] dF_{\phi }(u)\leq \frac{m_{G}}{\phi
_{N,m}\left( t\right) }\int_{-\infty }^{\infty }|u|dF_{\phi }(u), \label{th1-2}
\end{align}%
where the last passage follows from Assumption \ref{ass-5}\textit{(i)}, with 
$m_{G}$ an upper bound for the density of $G$. Also%
\begin{equation}
\int_{-\infty }^{\infty }\left( R^{1/2}\left[ G_{\phi }(u\phi
_{N,m}^{-1}\left( t\right) )-G_{\phi }(0)\right] \right) ^{2}dF_{\phi
}\left( u\right) \leq \frac{R}{\phi _{N,m}^{2}\left( t\right) }%
m_{G}\int_{-\infty }^{\infty }u^{2}dF_{\phi }(u). \label{th1-3}
\end{equation}%
Hence, using Assumptions \ref{ass-4} and \ref{restriction-1}, we conclude
via Markov's inequality that 
\begin{align*}
\Theta _{t}& =\int_{-\infty }^{\infty }\left\{ \frac{1}{\sqrt{G_{\phi }(0)%
\left[ 1-G_{\phi }(0)\right] }R^{1/2}}\sum_{i=1}^{R}\left[ I\{\xi _{i}\leq
0\}-G_{\phi }(0)\right] \right\} ^{2}dF_{\phi }(u)+o_{P^{\ast }}(1)= \\
& =\left\{ \frac{1}{\sqrt{G_{\phi }(0)\left[ 1-G_{\phi }(0)\right] }R^{1/2}}%
\sum_{i=1}^{R}\left[ I\{\xi _{i}\leq 0\}-G_{\phi }(0)\right] \right\}
^{2}+o_{P^{\ast }}(1),
\end{align*}%
and therefore the result follows from the Central Limit Theorem for
Bernoulli random variables. 

Equation (\ref{theta-alternative}) can be shown by exactly the same logic as the proof of (\ref{second-rand-2}), and of Theorem 4 in \citet{HT16}, and is therefore omitted.
\end{proof}

\begin{proof}[Proof of Theorem \protect\ref{gamma}]
The proof of the theorem is exactly the same as that of Theorem \ref{theta}, but this time 
based on Lemma \ref{second-randomisation}.
\end{proof}

\begin{proof}[Proof of Theorem \protect\ref{darling-erdos}]
By (\ref{exp-gamma}) and (\ref{var-gamma}), we have%
\begin{equation}
\max_{1\leq k\leq T_{m}}\sqrt{\frac{m}{k\left( k+m\right) }}\left\vert
\sum_{t=m+1}^{m+k}\frac{\Gamma _{t}-1}{\sqrt{2}}\right\vert =\max_{1\leq
k\leq T_{m}}\sqrt{\frac{m}{k\left( k+m\right) }}\left\vert
\sum_{t=m+1}^{m+k}Z_{t}\right\vert +O\left( m^{-\epsilon }\right) ,
\label{max-1}
\end{equation}%
where%
\begin{equation*}
Z_{t} = \frac{\Gamma _{t}-E^{\dag }\left( \Gamma _{t}\right) }{\sqrt{%
V^{\dag }\left( \Gamma _{t}\right) }}
\end{equation*}%
is an \textit{i.i.d.} sequence with mean zero, unit variance and finite
moments of order $2+\delta $. Consider (\ref{erdos-1}); on account of (\ref%
{max-1}), this holds immediately, following the same passages as in the
proof of Theorem 2.1 in \citet{lajos04}. As far as (\ref%
{erdos-2}) is concerned, due to the polynomial rate of approximation in (\ref%
{max-1}), it suffices to prove that%
\begin{equation*}
P^{\dag }\left( A_{m}\max_{1\leq k\leq T_{m}}\sqrt{\frac{m}{k\left(
k+m\right) }}\left\vert \sum_{t=m+1}^{m+k}Z_{t}\right\vert \geq
x+D_{m}\right) =e^{-e^{-x}},
\end{equation*}%
as $\min \left( N,m,R,W\right) \rightarrow \infty $. This is a relatively
standard exercise, and it is very similar to the proof of Theorem 1.1 in \citet{lajos07}; we therefore report only the main passages.
Let $a\left( m\right) =\left( \ln m\right) ^{2}$; by virtue of (\ref%
{liapunov-gamma}), it holds that%
\begin{equation}
\sup_{a\left( m\right) \leq k<\infty }\sqrt{\frac{m}{k\left( k+m\right) }}%
\left\vert \sum_{t=m+1}^{m+k}Z_{t}-B\left( k\right) \right\vert =O_{P^{\dag
}}\left( a\left( m\right) ^{-\frac{\epsilon }{2\left( 2+\epsilon \right) }%
}\right) ,  \label{kmt}
\end{equation}%
where $\left\{ B\left( t\right),\ 0\leq t<\infty \right\} $ is a
standard Wiener process -- see \citet{KMT1,KMT2}. We
now show that%
\begin{equation}
\max_{a\left( m\right) \leq k\leq \frac{cm}{\ln m}}\frac{B\left( \frac{k}{k+m%
}\right) }{\sqrt{\frac{k}{m}}}=\max_{a\left( m\right) \leq k\leq \frac{cm}{%
\ln m}}\frac{B\left( \frac{k}{k+m}\right) }{\sqrt{\frac{k}{k+m}}}+O_{P^{\dag
}}\left( \frac{\left( \ln \ln m\right) ^{1/2}}{\ln m}\right) ;
\label{modulus}
\end{equation}%
given that%
\begin{equation*}
\left\vert \frac{k}{m}-\frac{k}{k+m}\right\vert \leq \left( \frac{k}{m}%
\right) ^{2},
\end{equation*}%
using the modulus of continuity of the Wiener process we obtain%
\begin{equation*}
\max_{a\left( m\right) \leq k\leq \frac{cm}{\ln m}}\left\vert \frac{B\left( 
\frac{k}{k+m}\right) }{\sqrt{\frac{k}{m}}}-\frac{B\left( \frac{k}{k+m}%
\right) }{\sqrt{\frac{k}{k+m}}}\right\vert =O_{P^{\dag }}\left( 1\right)
\max_{a\left( m\right) \leq k\leq \frac{cm}{\ln m}}\frac{k}{m}\left( \ln 
\frac{m}{k}\right) ^{1/2},
\end{equation*}%
whence (\ref{modulus}) follows. Consequently, the following results hold:%
\begin{eqnarray}
&&\frac{1}{\sqrt{2\ln \ln m}}\max_{1\leq k\leq T_{m}}\frac{B\left( \frac{k}{%
k+m}\right) }{\sqrt{\frac{k}{m}}}\overset{P^{\dag }}{\rightarrow }1,
\label{lemma3.4} \\
&&A_{m}\max_{1\leq k\leq a\left( m\right) }\frac{B\left( \frac{k}{k+m}%
\right) }{\sqrt{\frac{k}{m}}}-D_{m}\overset{P^{\dag }}{\rightarrow }-\infty ,
\label{lemma3.5} \\
&&A_{m}\max_{\frac{cm}{\ln m}\leq k\leq T_{m}}\frac{B\left( \frac{k}{k+m}%
\right) }{\sqrt{\frac{k}{m}}}-D_{m}\overset{P^{\dag }}{\rightarrow }-\infty ;
\label{lemma3.5b}
\end{eqnarray}%
the results above are shown, for $\sqrt{\frac{k+m}{m}}B\left( \frac{k}{k+m}%
\right) $, in Lemmas 3.4, 3.5 and (in the proof of) Lemma 3.6 in \citet{lajos07}; in (\ref{lemma3.5b}) we have used the fact that, by
Assumption \ref{ass-6}, there exists a $c>0$ such that $T_{m}>cm$. Combining
(\ref{kmt}), (\ref{lemma3.4}), (\ref{lemma3.5}), (\ref{lemma3.5b}) and (\ref%
{modulus}) together, we obtain%
\begin{align*}
&P^{\dag }\left( A_{m}\max_{1\leq k\leq T_{m}}\frac{1}{\sqrt{k\left(
k+m\right) }}\left\vert \sum_{t=m+1}^{m+k}Z_{t}\right\vert \geq
x+D_{m}\right) =\\
&=P^{\dag }\left( A_{m}\max_{a\left( m\right) \leq k\leq \frac{cm}{\ln m}}%
\frac{\left\vert B\left( \frac{k}{k+m}\right) \right\vert }{\sqrt{\frac{k}{%
k+m}}}\geq x+D_{m}\right) +o\left( 1\right) ;
\end{align*}%
then the desired result follows from Lemma 3.6 in \citet{lajos07}.

Consider now (\ref{power-suffcond-2}); it is convenient to prove the result
under $H_{A,2}$ first. On account of (\ref{lambda-population}), Lemmas \ref%
{rolling-eigenvalues} and \ref{rolling-eigenvalues-2}, Assumption \ref{ass-4}%
\textit{(i)} and (\ref{gamma-alternative}), it holds that%
\begin{equation*}
\Gamma _{t}=C_{0}W+o_{P^{\dag }}\left( W\right), \text{ for }t\geq \tau
+C_{1}m^{1/2+\epsilon },
\end{equation*}%
where $\epsilon >0$ is such that $\frac{N^{1-\delta }}{m^{1/2-\epsilon }}%
\rightarrow C_{2}\in \left( 0,+\infty \right) $ and 
\begin{equation*}
C_{0}=\int_{-\infty }^{+\infty }\frac{\left\vert G_{\psi }\left( u\right)
-G_{\psi }\left( 0\right) \right\vert ^{2}}{G_{\psi }\left( 0\right) \left[
1-G_{\psi }\left( 0\right) \right] }dF_{\psi }\left( u\right) .
\end{equation*}%
Thus, standard algebra yields that under $H_{A,2}$%
\begin{equation*}
\sum_{m+1}^{m+k}\frac{\Gamma _{t}-1}{\sqrt{2}}=O_{P^{\dag }}\left( 1\right) %
\left[ m+k-\left( \tau +C_{1}m^{1/2+\epsilon }\right) \right] W+o_{P^{\dag
}}\left( W\right) ,
\end{equation*}%
whenever $k\geq \tau +C_{1}m^{1/2+\epsilon }$. Therefore, 
\begin{equation*}
\Lambda _{m}=O_{P^{\dag }}\left( 1\right) m^{1/2}W\max_{1\leq k\leq T_{m}}%
\frac{m+k-\left( \tau +C_{1}m^{1/2+\epsilon }\right) }{k^{1/2}\left(
m+k\right) ^{1/2}}+o_{P^{\dag }}\left( W\right) ;
\end{equation*}%
elementary algebra yields 
\begin{equation*}
\max_{1\leq k\leq T_{m}}\frac{m+k-\left( \tau +C_{1}m^{1/2+\epsilon }\right) 
}{k^{1/2}\left( m+k\right) ^{1/2}}\geq C_{2}>0;
\end{equation*}%
thus, Assumption \ref{ass-6} implies (\ref{power-suffcond-2}). Under $%
H_{A,1} $ the logic is similar, and therefore only the main passages are
reported. Under $H_{A,1}$ we have%
\begin{equation*}
\Gamma _{t}=C_{0}W+o_{P^{\dag }}\left( W\right), \text{ for }\tau
+C_{1}m^{1/2+\epsilon }\leq t\leq \tau +m-C_{1}m^{1/2+\epsilon },
\end{equation*}%
with the same notation as above; it is then easy to see that%
\begin{equation*}
\max_{1\leq k\leq T_{m}}\frac{m+k-\left( \tau +C_{1}m^{1/2+\epsilon }\right) 
}{k^{1/2}\left( m+k\right) ^{1/2}}\geq C_{2}>0;
\end{equation*}%
the proof is now the same as before.
\end{proof}

\begin{proof}[Proof of Corollary \protect\ref{corollaryerdos}]
The corollary is an immediate consequence of Theorem \ref{darling-erdos} and
its proof. Considering (\ref{size}), note that $P^{\dag }\left( \widehat{%
\tau }_{m}<T\right) $ is monotonically nondecreasing in $T$; by
definition%
\begin{align*}
P^{\dag }\left( \widehat{\tau }_{m}<T\right)  &=P\left( \max_{1\leq
k\leq T}\frac{d\left( k;m\right) }{\nu ^{\ast }\left( k;m\right) }%
>c_{\alpha ,m}\right) \leq P\left( \max_{1\leq k<\infty }\frac{d\left( k;m\right) }{\nu ^{\ast
}\left( k;m\right) }>c_{\alpha ,m}\right) = \\
&=P\left( \sup_{0\leq t\leq 1}\frac{\left\vert B\left( t\right) \right\vert 
}{t^{\eta }}>c_{\alpha ,m}\right) +o\left( 1\right) =\alpha ,
\end{align*}%
which proves (\ref{erdos-1}); (\ref{erdos-2}) follows from the same
passages. Similarly, as far as (\ref{power}) is concerned, note that%
\begin{equation*}
P^{\dag }\left(t_{N,m}^{\ast }\le  \widehat{\tau }%
_{m}<T\right)  =P\left( c_{\alpha
,m}^{-1}\max_{1\leq k\leq T_{m}}\frac{d\left( k;m\right) }{\nu ^{\ast
}\left( k;m\right) }>1\right) =1,
\end{equation*}%
by (\ref{power-suffcond-2}).
\end{proof}

\end{document}